\documentclass[letter,scriptaddress,twocolumn, showkeys]{revtex4}
\usepackage{amssymb}
	\usepackage{amsmath,amsthm}
	\usepackage{makeidx}
	\usepackage{amsfonts}
	\usepackage[ansinew]{inputenc}
	\usepackage[usenames,dvipsnames]{pstricks}
	\usepackage{courier}
	\usepackage{mathtools}
	\usepackage{subfigure}
	\usepackage{epsfig}
	\usepackage{pst-grad} 
	\usepackage{pst-plot} 
	\usepackage[colorlinks,hyperindex]{hyperref}
	\usepackage{subfigure}

	\hypersetup
	{
		colorlinks,%
		citecolor=blue,%
		linkcolor=blue,%
		urlcolor=blue,%
	}

	\newtheorem{proposition}{Proposition}
	\newtheorem{definition}{Definition}
	
 	\newtheorem{lemma}{Lemma}

\makeindex

\begin{document}

\title{Modeling epidemics on $d$-cliqued graphs}

\author{Laura P. Schaposnik~$^{(a)(b)}$ and  Anlin Zhang~$^{(c)}$}

    \affiliation{(a) University of Illinois, Chicago, IL 60607, USA.\\~
    (b) Freie   Universit\"at Berlin, 14195 Berlin, Germany. 
\\~
    (c) Canyon Crest Academy, San Diego, CA 92130, USA.~}



\begin{abstract}
Since social interactions have been shown to lead to symmetric clusters,  we propose here that symmetries  play a key role in epidemic modeling.  Mathematical models on $d$-ary tree graphs were recently shown to be particularly effective for  modeling epidemics in simple networks [Seibold $\&$ Callender, 2016].  To account for symmetric relations, we generalize this to a new type of networks modeled on $d$-cliqued tree graphs, which are obtained by adding edges to regular $d$-trees to form $d$-cliques.  This setting gives a more realistic model for epidemic outbreaks originating, for example,  within a family or classroom and which could reach a population by transmission via children in schools.  Specifically, we quantify how an infection starting in a clique (e.g. family) can reach other cliques through the body of the graph (e.g. public places).    \end{abstract}

\maketitle



\section{Introduction}

The study of epidemic propagation in networks (social and biological) has been of much interest to biologists and mathematicians for a long time, but only recently have graph theory, number theory, and computer science taken researchers to several breakthroughs. Moreover, as highlighted in \cite{Pastor}, the importance of local clustering in networks has been widely recognised, and thus understanding the connectivity of the networks is of upmost importance for developing effective control measures (e.g., quarantine, vaccinations or specific treatments). 

The mathematical foundation for epidemiology can be traced to the early 1900s with the work of Ronald Ross: he produced the first mathematical model of mosquito-borne pathogen transmission using mosquito spatial movement in order to reduce malaria from an area \cite{Ross}. In the 1920s, Lowell Reed and Wade H. Frost developed the \textit{Reed-Frost model} \cite{Reed}, which improved on Ross's model by modeling how an epidemic behaves with respect to time. In 1927, William O. Kermack and Anderson G. McKendrick \cite{Kermack} created the \textit{SIR model}, which categorized people into the $3$ states {\it Susceptible}, {\it Infectious} and {\it Removed}. More recently, Matt J. Keeling and Ken T.D. Eames used contact networks to better represent a community \cite{Keeling}.  
The use of contact networks which are adapted to reflect certain particular characteristics of society has been of much use when doing mathematical modeling of epidemics. This modeling is performed by seeing the network as a graph where vertices represent individuals, and edges encode the interactions amongst people: 
two people, seen as vertices, are connected by an edge in the graph whenever they are related (and thus an interaction could exist)

Particular shapes of contact networks have been studied in recent years, from  grid contact network to represent fields and study fungal infections, to networks with regular $d$-ary tree structure,  used  to study SARS outbreaks in Hong Kong (e.g., see \cite{Riley}). Moreover,  in \cite{Seibold}  the authors  considered the implications of disease spread on perfect $d$-ary trees. Our work  builds on their research, extending their results to more general networks and index cases.

We dedicate this paper to study local clustering in networks and its impact on epidemic modeling by incorporating an important method: the analysis of the {\it symmetries} which networks have through the appearance of cliques.
 This perspective is inspired by the appearance of symmetric relations among members of subgraphs within a network. An overall question we aim at understanding is {\it  how an epidemic outbreak originating within a family or classroom could reach a population by transmission via children in schools}.  
 To answer these questions, we introduce a novel type of graphs, the {\it $d$-cliqued trees}, to study  epidemics in regular networks which contain symmetric clusters both through next generation models as well as through general discrete-time models, which allows one to assign probabilities of infection $P_{inf}$ and recovery $P_{rec}$ to track how an outbreak would affect the population modeled via our contact network.

\begin{figure}[h]
\centering
\subfigure{%
\resizebox*{6cm}{!}{\includegraphics{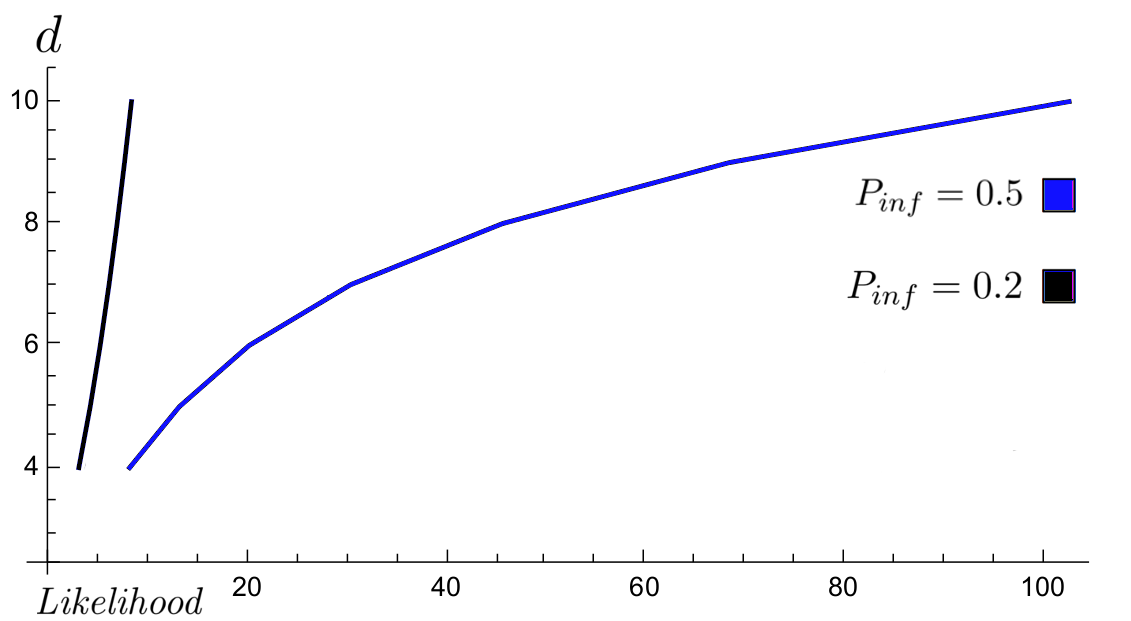}}}
\caption{The quotient of the probability of an outbreak lasting two times in our model divided by the same probability for the model in  \cite{Seibold}, more details in Figure \ref{like}.} \label{like1}
\end{figure}

\vspace{-0.5 cm}

Through $d$-cliqued trees $\mathcal{A}_d^\lambda$   of degree $d$ and height $\lambda$ introduced in Definition \ref{defi1}, in what follows we investigate how the outbreak duration of an epidemic originated within a classroom propagates through society depending on the size of the contact network considered (the height $\lambda$), as well as on the size of the classroom (the degree $d$). 
 Since we aim at highlighting the importance of symmetries to approach these realistic settings, we will  elaborate on how our probabilities $P_i$ of outbreaks lasting $i$ time units (ticks) differ from those in    \cite{Seibold} -- see, for example    Figure \ref{like1}.

\section{A model on $d$-ary trees}

Symmetries should appear in networks when considering clusters of vertices which share some common relation (for example, the members of a family and their  symmetric relations, or  cities connecting to a common airport). In 2016, researchers studied epidemic modeling on $d$-ary trees  \cite{Seibold} using the  IONTW platform: in this paper we present  a new approach to generalize the study to \textit{d-cliqued tree graphs}, given by regular trees with  added edges creating cliques, or symmetric and completely connected subnetworks. We dedicate this section to give  a brief overview of regular trees and $d$-ary trees in Section \ref{background}, and of the main results of     \cite{Seibold} in Section \ref{Seibold}.

\subsection{Background: regular trees and $d$-ary trees}\label{background}

In modeling epidemics on simple graphs, vertices denote  individuals,  and   edges the possible ways to transmit an infection from one individual to another.  Following the notation of  \cite{Seibold}, we refer to the \textit{index node}  $I$ of the graph as the vertex where  the infection originates.  
\begin{definition}
A \textbf{tree} is an undirected graph in which any two vertices are connected by exactly one path. 
A \textbf{$d$-regular tree} is a tree graph for which  all non-terminal vertices have degree $d$ (i.e., have $d$ adjacent vertices), and the terminal ones have degree 1. The degree  of a $d$-regular tree  is $d$. 
\end{definition}

A {\it rooted tree} is a tree in which one vertex $x_0$ has been designated the {\it root}. In such trees, the {\it parent} of a vertex $v$ is the vertex connected to it on the path to the root $x_0$, and $v$ is  called a {\it child} of the parent vertex. The {\it height} is the length, in number of edges, from a terminal vertex to the root. We assume that all terminal vertices of the $d$-regular tree have the same height $\lambda\in \mathbb{N}$, and call this the {\it height} of the $d$-regular tree. The {\it network diameter}  $L$ of a $d$-regular tree is the number of edges in the longest path between two vertices. Regular trees are closely related to $d$-ary trees (Figure \ref{example2}):

\vspace{-5 mm}
 \begin{figure}[h]
\centering
\subfigure[A regular $3$-tree  with root the centermost vertex.]{%
\resizebox*{3.5 cm}{!}{\includegraphics{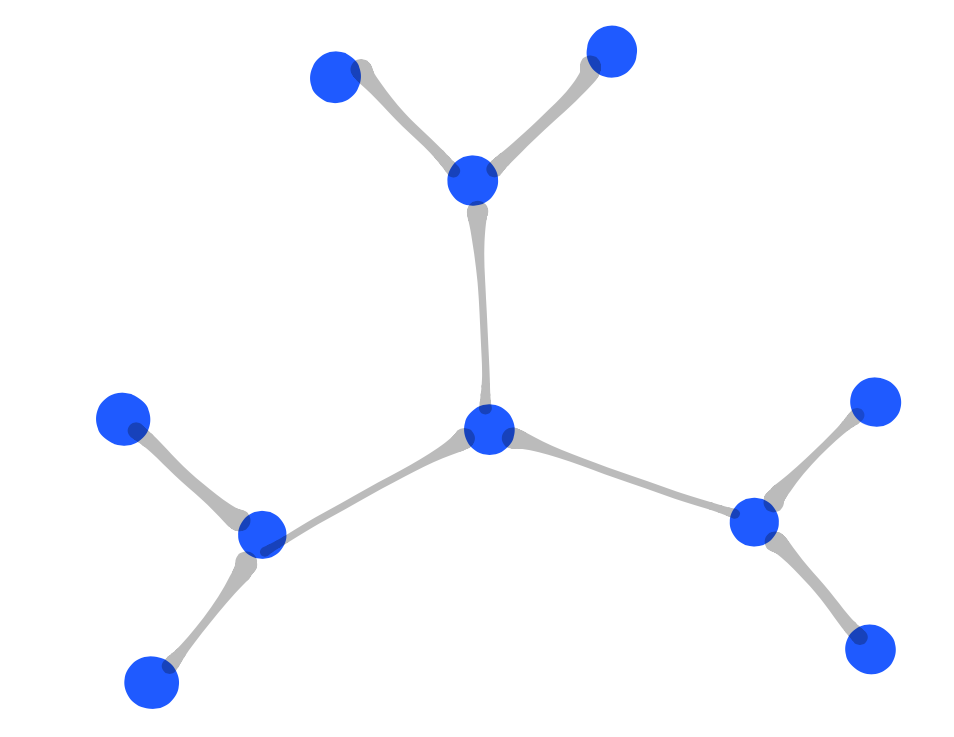}}}\hspace{15pt}
\subfigure[A $5-$ary tree with root the centermost vertex.]{%
\resizebox*{3.5 cm}{!}{\includegraphics{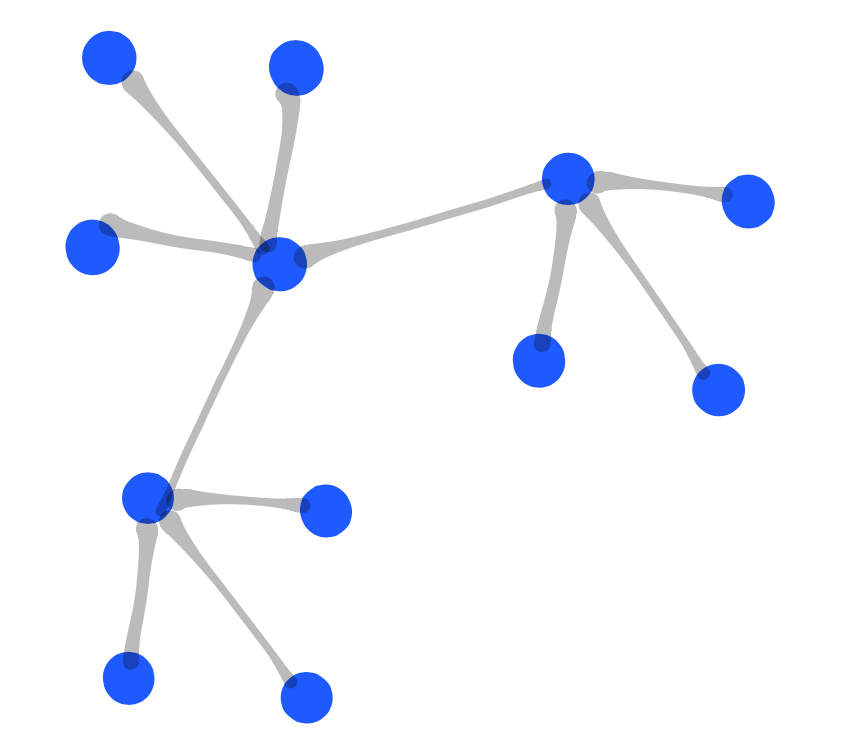}}}
\caption{Comparison of two types of trees. Figure (b) gives a network studied in \cite{Seibold}.} \label{example2} 
\end{figure}
\begin{definition} A \textbf{$d$-ary tree} is a rooted tree in which each node has at most $d$ children. A \textbf{perfect $d$-ary tree} is one where  non-terminal nodes have   $d$ children, and  terminal nodes have none. 
\end{definition}

 It is known that a $d$-regular tree on $n$ vertices exists if and only if $d - 1$ divides $n - 2$.  Finally, for any graph $G$, we denote by $d(v,w)$ the length of the shortest path between two vertices $v,w\in G$, which will become useful in later sections (e.g., for Lemma \ref{lemma10}).

\subsection{Outbreak modeling  through the IONTW platform}\label{Seibold}

The Infections On NeTWorks platform (IONTW) \cite{Just,Wilensky} is an agent-based modeling platform that shows theoretical predictions for disease spread and simulates disease transmission on  graphs (see Figure \ref{sample-figure}). In  \cite{Seibold}, the authors studied how outbreak duration of a given disease is affected by the size of a given contact network whose structure resembles perfect $d$-ary trees via IONTW. It is important to note that whilst the authors referred to the networks as ``regular tree graphs'', the notion of regular graph (where every vertex has the same incidence degree) was not considered. Hence, the correct terminology for their networks would seem to be indeed the one of ``perfect $d$-ary trees''.

\begin{figure}[h]
\centering
\subfigure{%
\resizebox*{8.2 cm}{!}{\includegraphics{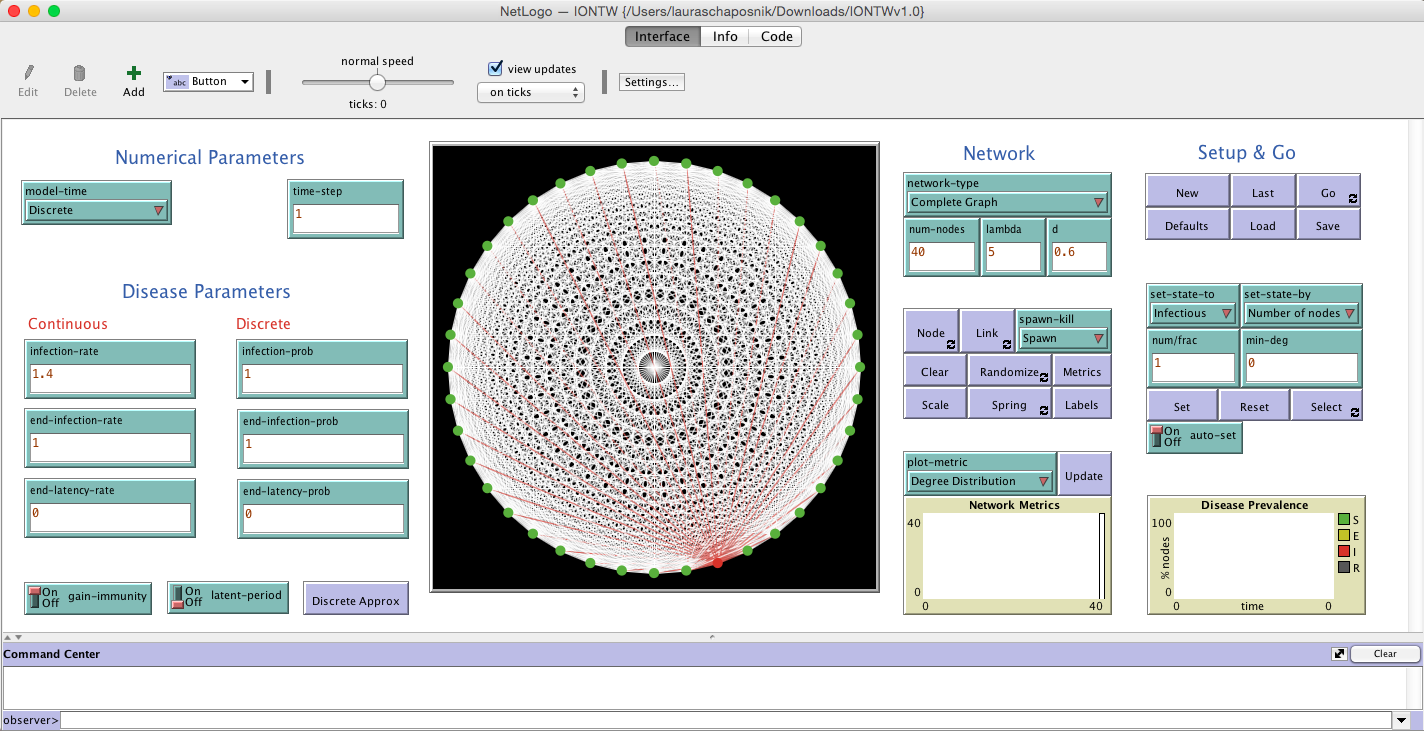}}}
\caption{Screenshot of the IONTW platform, which  can be used to model epidemic outbreaks for certain types of graphs. In the figure, a complete graph in 40 vertices.} \label{sample-figure}
\end{figure}

In  \cite{Seibold}, the authors initially assumed homogeneity of hosts, i.e., that the probabilities of infection \textbf{$P_{inf}$} are identical for all hosts in the population, as are the probabilities of recovery \textbf{$P_{rec}$}. The {\it outbreak duration} $\delta$ is the number of ticks, or units of time,  from the introduction of the infection in the index  node until all hosts have recovered. We will denote by $\delta_\lambda$ the expected average outbreak duration, and by $P_i$ the probability the outbreak lasts exactly $i$  ticks.

In a next generation model, one assumes   $P_{rec} = 1,$ i.e., infected hosts are moved to the recovered group at some time step, and  $P_{inf}=1,$ i.e., that  all nodes will eventually get infected. The lower bound for $\delta_\lambda$ on regular  tree graphs with $d>1$ is $\lambda+1$   when the   origin $I$ is placed in the root, since the maximum distance from there  to any node is $\lambda$, and we need an additional time step for the last node to recover.  Hence, in this case when the index is placed in the root, by definition of expected value, one has $\delta_\lambda:=\sum_{i=1}^{\lambda+1} i\cdot P_i.$
\pagebreak

 Finally, it was shown in  \cite{Seibold} that the probabilities $P_i$ of the duration spanning $i$ time units, when the origin $I$ is the root of a perfect $d$-ary tree,  is given by 
\begin{eqnarray}
P_1 &= & (1-P_{inf})^d\label{probabi1}\\
 P_2 &= &\sum_{r=1}^{d} \binom{d}{r} P_{inf}^{r} (1-P_{inf})^{d-r}(1-P_{inf})^{d\cdot r} \label{p2tree}\\
& \vdots &\nonumber\\
P_{\lambda + 1} &=& 1-(P_1+P_2+\cdots +P_{\lambda}).\label{probabi2}
\end{eqnarray}
For each $P_i$ term, the outer sum denotes nodes adjacent to the infection's origin $I$ which is placed in the tree's root, whilst the inner sum denotes nodes that are infected $i-1$ edges away from the infection's origin. Moreover, from  \cite{Seibold}, the outbreak duration $\delta$ is bounded as follows:  \begin{eqnarray}
 \lambda +1\leq &\delta& \leq 2\lambda+1~{\rm ~when~}~ d>1, \\
 \frac{\lambda}{2} \leq &\delta& \leq \lambda + 1~{\rm ~when~}~  d=1.
 \end{eqnarray}

Much of the work in  \cite{Seibold} was also dedicated to the study of the change a network went through when the degree $d$ or height $\lambda$ of the graph varied. In particular, through the IONTW platform, the authors obtained the following descriptions in Figure \ref{FigureOld}, which should be compared with the ones in Section \ref{masclique} which we obtain without the use of IONTW.

 \begin{figure}[h]
\centering
\subfigure[Number of nodes $N$ in a $d$-ary graph generated in IONTW for various values of $\lambda$ and $d$.]{%
\resizebox*{8.2cm}{!}{\includegraphics{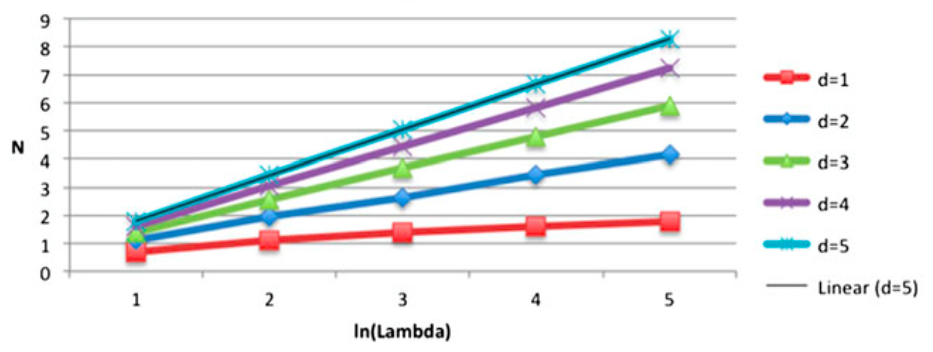}}}\hspace{5pt}
\subfigure[Semi-log plot of number of nodes $N$ in a $d$-ary tree graph generated in IONTW.]{%
\resizebox*{8.2cm}{!}{\includegraphics{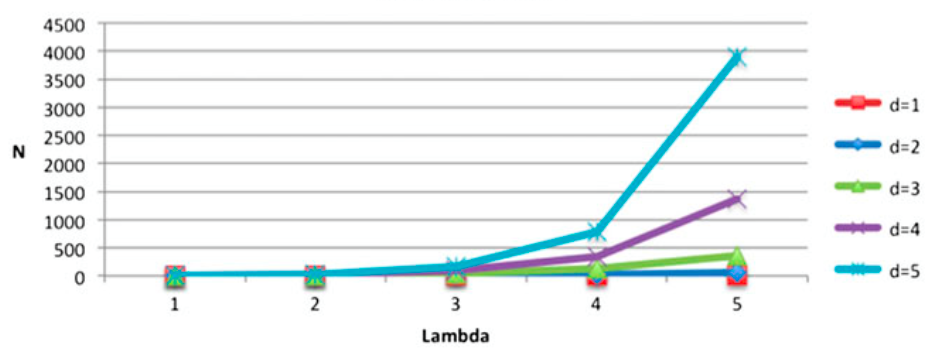}}}
\caption{Figures from  \cite{Seibold} obtained through IONTW on the network's  size variation.} \label{FigureOld} 
\end{figure}
Through the understanding of the probabilities $P_i$ the authors looked into the duration of outbreaks depending on the placement of the index node, as well as on the size of the network, although an explicit expression of $P_i$ was only given when $I$ was taken to be the root of the graph. Through a next-generatoin model where $P_{inf}$ is set to $1$, and when the index case is the root of the graph, the disease duration was described as a function of population size,  determined by varying values of the height $\lambda$  for degree $d=2$ networks as in Figure \ref{FigureOld2} (a), and for different values of $d$ with the same height $\lambda$ as in Figure \ref{FigureOld2} (b).  
 
 \begin{figure}[h]
\centering
\subfigure[Disease duration in terms of network size for randomly chosen index case using IONTW.]{%
\resizebox*{8.2cm}{!}{\includegraphics{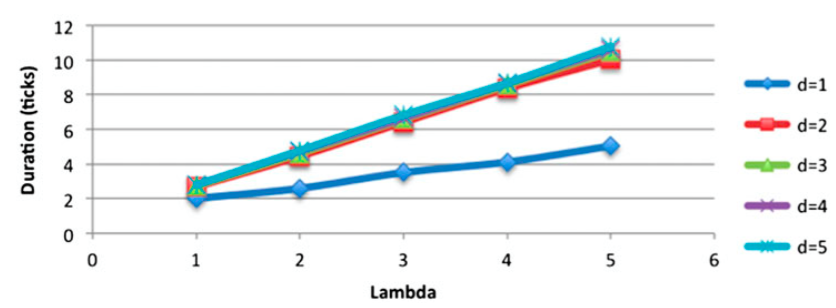}}}\hspace{5pt}
\subfigure[Disease duration in terms of network height.]{%
\resizebox*{8.2cm}{!}{\includegraphics{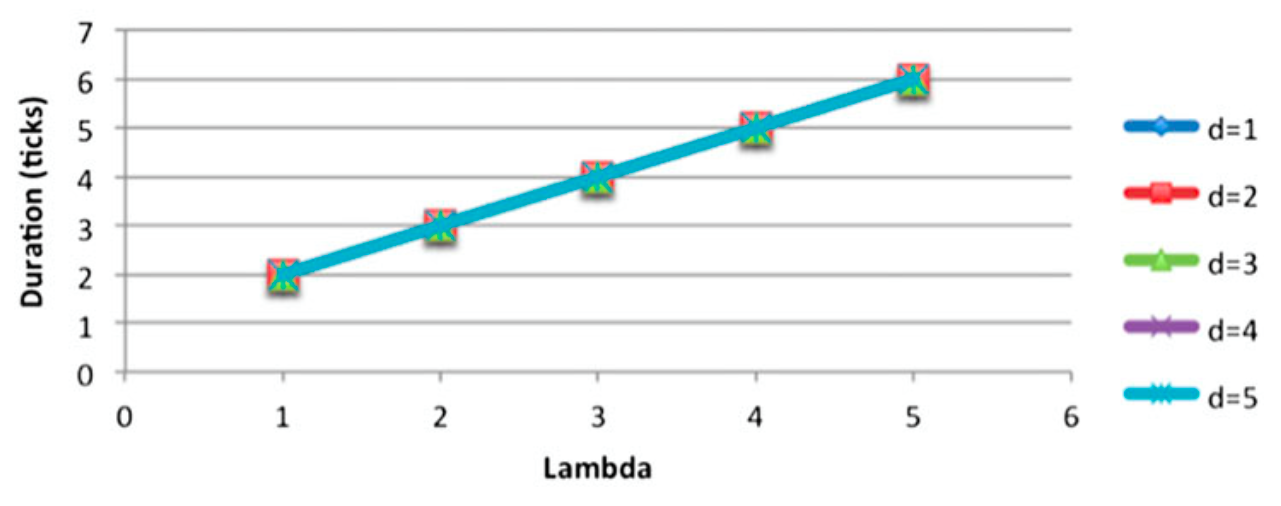}}}
\caption{Figures from  \cite{Seibold} obtained through IONTW  on disease durations.} \label{FigureOld2} 
\end{figure}

In particular  it was observed  that the duration will persist for $\lambda + 1$ ticks regardless of $d$ when the index case is located in the root of the graph  \cite{Seibold}, and   formulas for the mean duration $\delta_\lambda$ of a disease outbreak were provided for such case:
\begin{eqnarray}\label{mean}
\delta_\lambda=\left\{\begin{array}{lc}
\frac{1}{\sum_{i=0}^{\lambda}d^i}\sum_{j=0}^{\lambda}\left((d^j)({\rm Max}(\lambda-j,j)+1)\right)&d=1;\\
\frac{1}{\sum_{i=0}^{\lambda}d^i}\sum_{j=0}^{\lambda}\left((d^j)(\lambda+j+1)\right)&d>1.\end{array}\right. 
\end{eqnarray}

\section{Geometry of $d$-cliqued networks}\label{masclique}
 Our primary research question relates to how one can model epidemic outbreaks in networks that resemble family trees where the relations between siblings are taken into account. We shall begin by introducing these new types of networks, modeled on {\it $d$-cliqued graphs}, in Section \ref{dcliqued}. Since  we are interested in understanding how the duration of a given disease is affected by the size of  a {\it $d$-cliqued network}, we proceed in Sections \ref{scale}-\ref{diameter} to study how the scaling of  the  height $\lambda$ and  the degree $d$ affect the network. 

\subsection{A new type of network modeled on $d$-cliqued graphs}\label{dcliqued}

 Since symmetries in social networks can sometimes be modeled through cliques, in what follows we will  extend the work of   \cite{Seibold} to a more general setting of regular trees with added   cliques. 

\begin{definition}
A  \textbf{clique} is a subset of a graph such that all pairs of vertices in the subset are connected by an edge. A clique with $d$ vertices is called a $d$-clique.
\end{definition}
\begin{definition}
The \textbf{clique number} $\omega(G)$ of $G$ is the size of the largest clique in $G$, and a \textbf{maximum clique} is a clique with  $\omega(G)$ vertices.
\end{definition}

Cliques are especially useful in modeling epidemics, as they can represent completely related social groups (e.g. families or classrooms).  We will look at regular tree graphs with added cliques and analyze how duration, $P_{inf}$ and $P_{rec}$ can be expressed in terms of different variables.  
\begin{definition}\label{defi1}
A \textbf{$d$-cliqued tree graph } $\mathcal{A}_{d}^\lambda$ {\bf of height $\lambda$}  is a regular tree graph of degree $d$ and height $\lambda$, with added edges to form $d$-cliques on the terminal vertices.   The \textbf{body} of $\mathcal{A}_{d}^\lambda$  is  the set of  vertices which are not in any clique. 
\end{definition}

\begin{figure}[h]
\centering
\subfigure{%
\resizebox*{6cm}{!}{\includegraphics{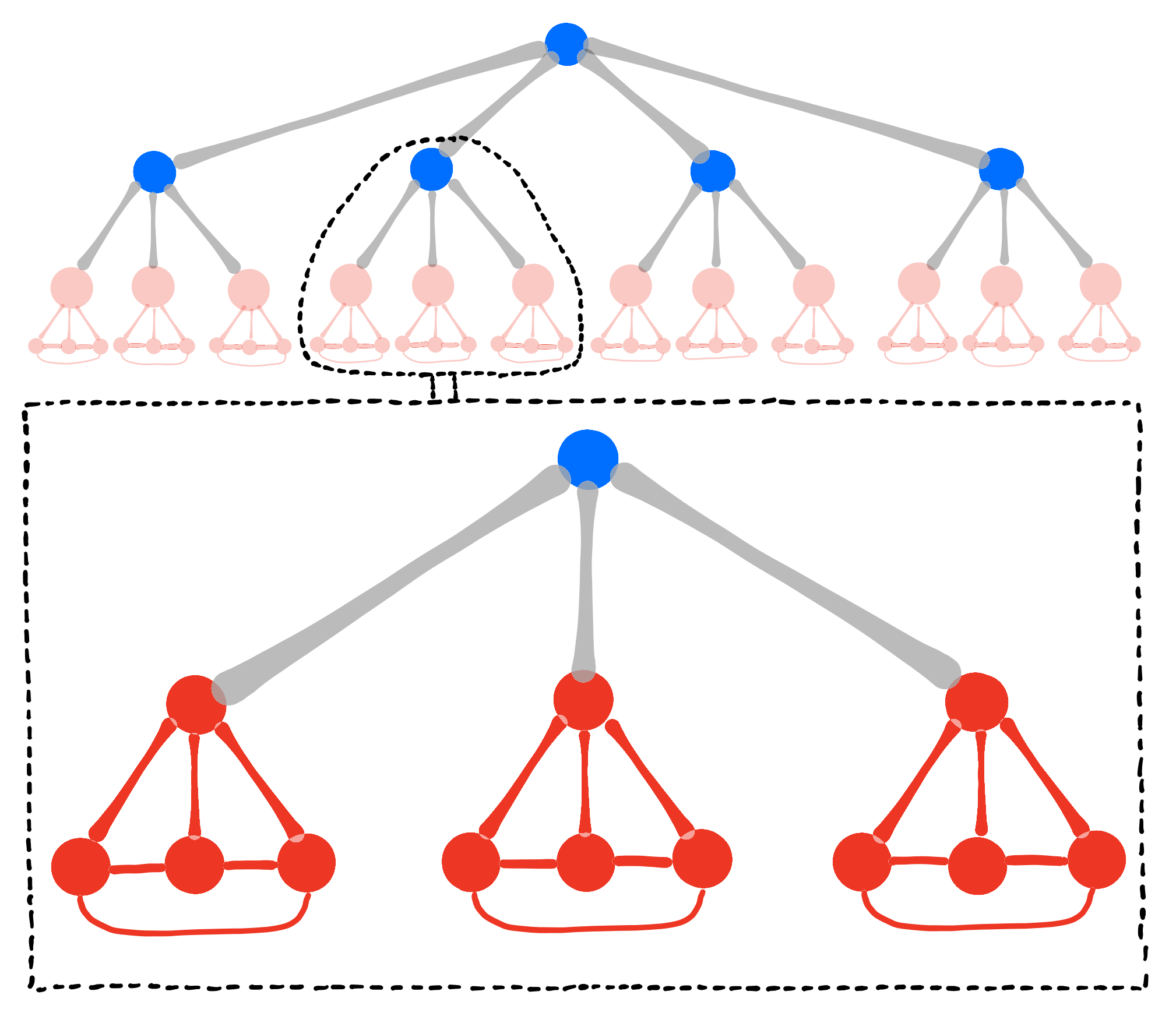}}}
\caption{Graph $\mathcal{A}^3_4$: a 4-regular tree  completed to have terminal  $4$-cliques (in red), and body vertices (in blue).} \label{hola}
\end{figure}

By construction, the number of $d$-cliques is $d (d-1)^{\lambda-2}$, and $\omega(\mathcal{A}_{d}^\lambda)=d$.  The graph $\mathcal{A}_d^{\lambda}$  can be separated into levels, where the root makes up level 0, adjacent nodes are contained in level 1, and so on, until the leafs make up level $\lambda$, the level where cliques are formed. 
  
\subsection{Size scaling with $\lambda$ and $d$}\label{scale}

Since $d$-cliqued networks are modeled on modifications of $d$-graphs to which one has added edges, the size of the network will depend on the height $\lambda$ and the degree $d$ of the underlying regular tree. In particular, we can show the following. 
\begin{proposition} The total number of edges of $\mathcal{A}_{d}^\lambda$ is 
\begin{align}
d^2(d-1)^{\lambda-2}\cdot\frac{d-2}{2}+\sum_{i=0}^{\lambda-1}d(d-1)^{i}.\label{edgenumber}
\end{align}
\end{proposition}

\begin{proof} 
Since the number of vertices of a $d$-regular tree of height $\lambda$ is 
$1+\sum_{i=0}^{\lambda-1}d(d-1)^i,$
 this is also the number of vertices of a $d$-cliqued tree  $\mathcal{A}_{d}^\lambda$. 
Moreover, the number of edges in a $d$-regular tree of height $\lambda$ is $\sum_{i=0}^{\lambda-1}d(d-1)^{i}$. Hence, since $d$-cliques have $d(d-1)/2$ edges,  the number of edges that one needs to add to a $d$-regular tree of height $\lambda$ to obtain the $d$-cliqued graph $\mathcal{A}_{d}^\lambda$ is 
$$\underbrace{  \left(\frac{d(d-1)}{2}-(d-1)\right)}_{{\rm edges~ missing~ to~ make~ a~ complete~}d-{\rm clique}}\cdot\overbrace{d(d-1)^{\lambda-2}}^{{\rm number~of~cliques~made}}$$ and thus the total number of edges of $\mathcal{A}_{d}^\lambda$ is 
as in   \eqref{edgenumber}. \end{proof}

\subsection{Clique and diameter scaling with $\lambda$}\label{diameter}

In what follows we shall study the dependence of the size of the network on the height $\lambda$ and the degree $d$ of the graphs $\mathcal{A}_{d}^\lambda$.  Note that from its definition, the number of vertices in the body of $\mathcal{A}_{d}^\lambda$  is $1+\sum_{i=0}^{\lambda-3}d(d-1)^i.$ In order to study the cliques of $\mathcal{A}_{d}^\lambda$, note that these have been formed through the outermost vertices, and one will have a clique for every vertex at the $\lambda-1$ level of the graph. Since the first level has $d$ vertices but subsequent ones increase by a factor of $(d-1)$, the number of cliques in the whole graph is $d(d-1)^{\lambda-2}.$ Then, one can understand the growth of the number of cliques in $\mathcal{A}_{d}^\lambda$ in terms of  $\lambda$ as shown in Figure \ref{holamas2} (a). 
 Moreover,  the number of edges in $\mathcal{A}_{d}^\lambda$ can also be seen in terms of   the degree $d$ and height $\lambda$ via  Eq. \eqref{edgenumber}, as illustrated in Figure \ref{holamas2} (b).

 \begin{figure}[h]
\centering
\subfigure[ ~The number of cliques in $\mathcal{A}_{d}^\lambda$.]{%
\resizebox*{8.2cm}{!}{\includegraphics{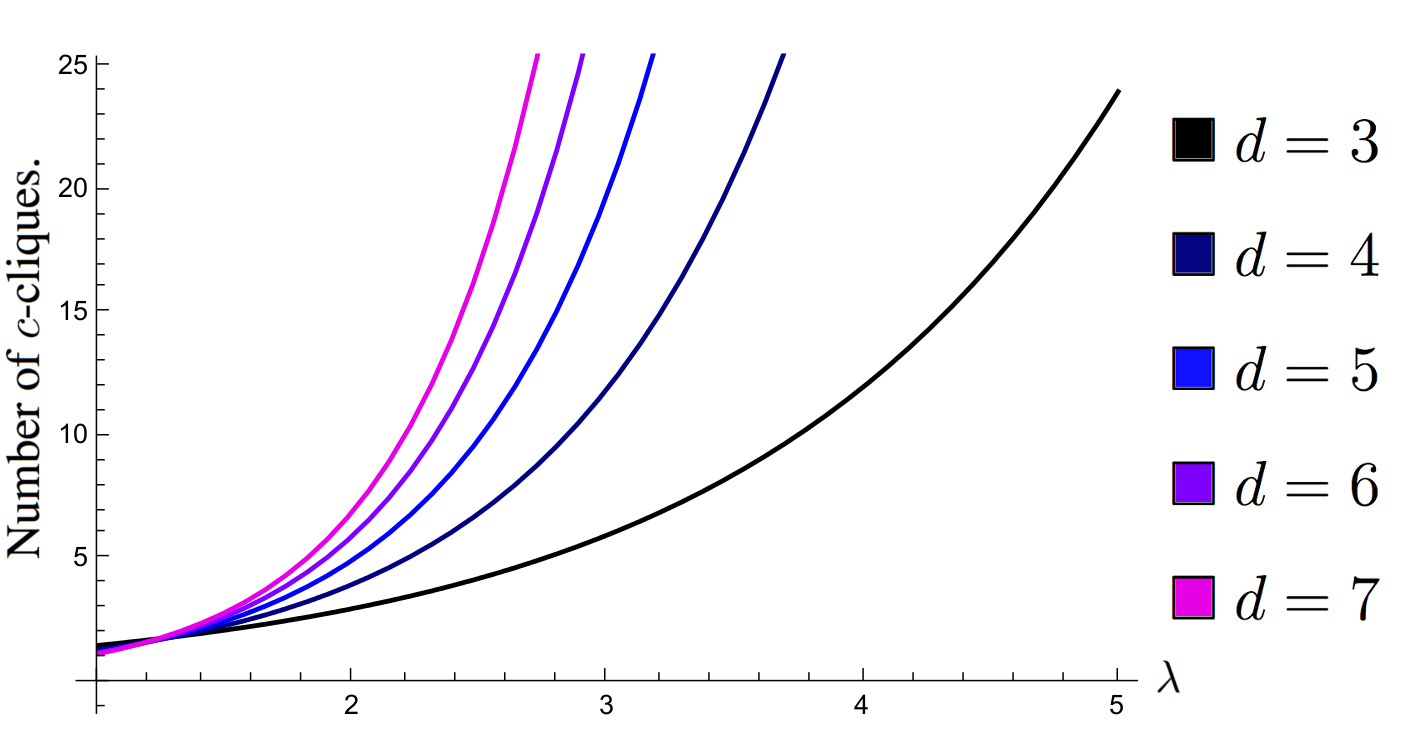}}}\hspace{5pt}
\subfigure[ ~The number of edges in $\mathcal{A}_{d}^\lambda$.]{%
\resizebox*{8.2cm}{!}{\includegraphics{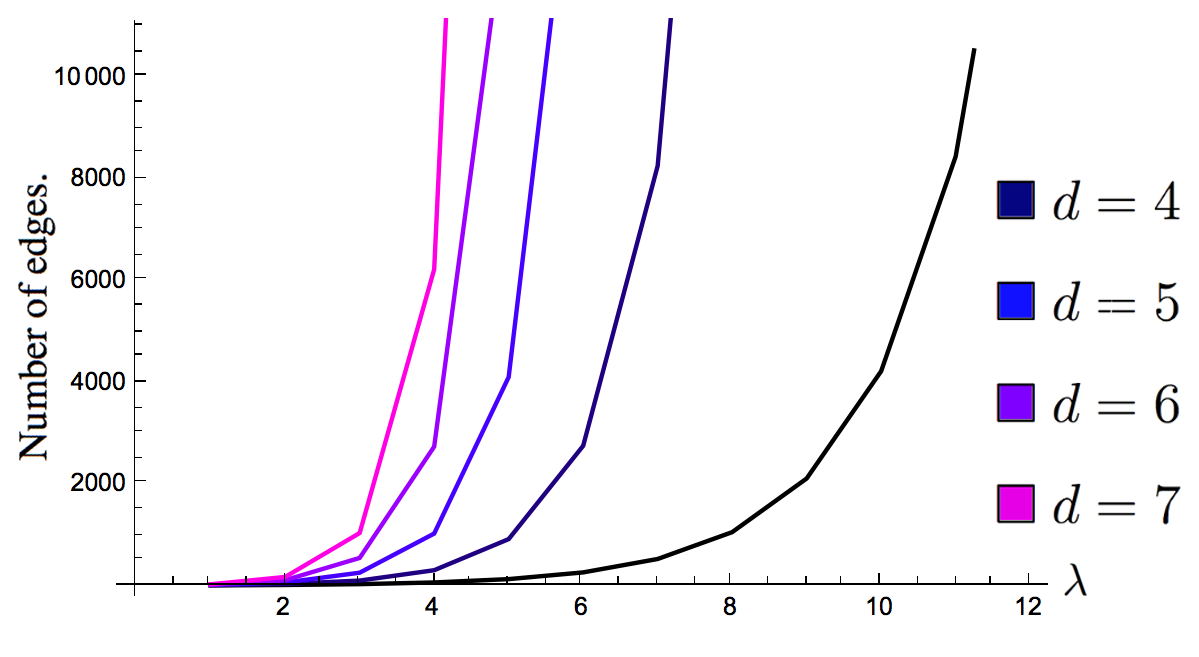}}}
\caption{Clique number and edge growth of $\mathcal{A}_{d}^{\lambda}$ in terms of the height $\lambda$.} \label{holamas2} 
\end{figure}

\subsection{Distances between vertices}\label{cases}

When studying epidemics in $d$-cliqued trees, one needs to consider the different positions that the index vertex (origin of the infection) may take. These positions may be classified into two cases:

\begin{itemize}
\item\textbf{{Case (I):}}   when the infection's origin  $I$ is adjacent to the body of $\mathcal{A}_{d}^\lambda$, as in Figure \ref{casesroot} (a), or is in the body of $\mathcal{A}_{d}^\lambda$,  as in   Figure \ref{casesroot} (b). In this case,  the shortest path between $I$ and any vertex  is in the underlying regular tree  and so the results of  \cite{Seibold} hold.
 \smallbreak 
 \item \textbf{{Case (II):}} when the infection's origin $I$ is in a clique and adjacent to only vertices in a clique as in Figure \ref{casesroot} (c).    Note that from its definition, the body of $\mathcal{A}_{d}^\lambda$  has $$1+\sum_{i=0}^{\lambda-3}d(d-1)^i$$ vertices, and as seen before the whole graph has $d(d-1)^{\lambda-2}$ cliques. The number of cliques in $\mathcal{A}_{d}^\lambda$ depending $\lambda$ is shown in Figure \ref{holamas2} (b). 

\end{itemize}

\begin{figure}[h]
\centering
\subfigure{%
\resizebox*{8.2cm}{!}{\includegraphics{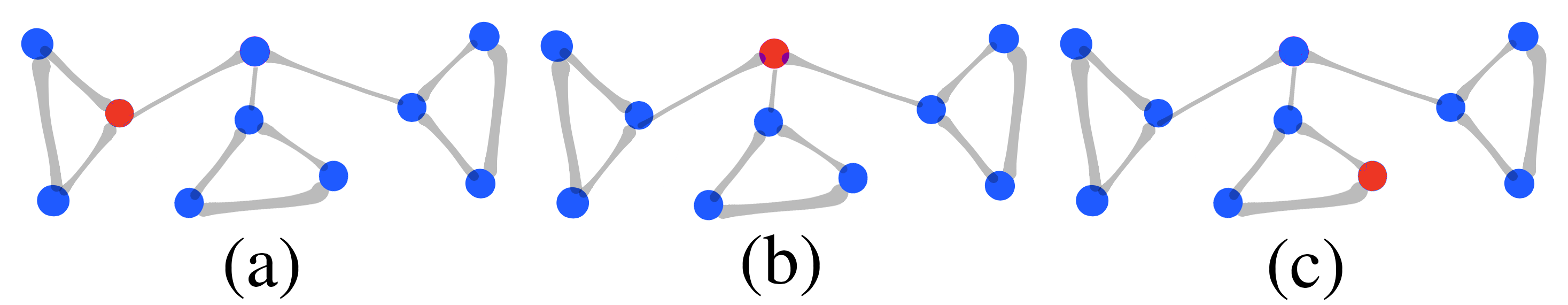}}}
\caption{Placements of an infection's origin (red) in  $\mathcal{A}_3^2$. Figures (a),(b)= Case (I); Figure (c)= Case (II).} \label{casesroot}
\end{figure}

In order to study durations of disease outbreaks on $d$-cliqued networks, one first needs to understand the distances between nodes and the infection index, as well as how many nodes are at each distance. In what follows we shall answer these questions.

\begin{definition} Denote by $D_x$ the shortest path from a vertex $x$ to the root of the underlying tree of $\mathcal{A}_{d}^\lambda$. Its size $|D_x|$ is the {\bf distance} from $x$ to the root.
\end{definition}

From the above analysis one can see that the most interesting case for us will be when the index case is in the outermost leaf of the network, and hence in what follows we shall fix such an $I$. That is, an index case $I$ such that $|D_I|=\lambda$ as in Figure \ref{yo25}. 
\begin{figure}[h]
\centering
\subfigure{%
\resizebox*{8cm}{!}{\includegraphics{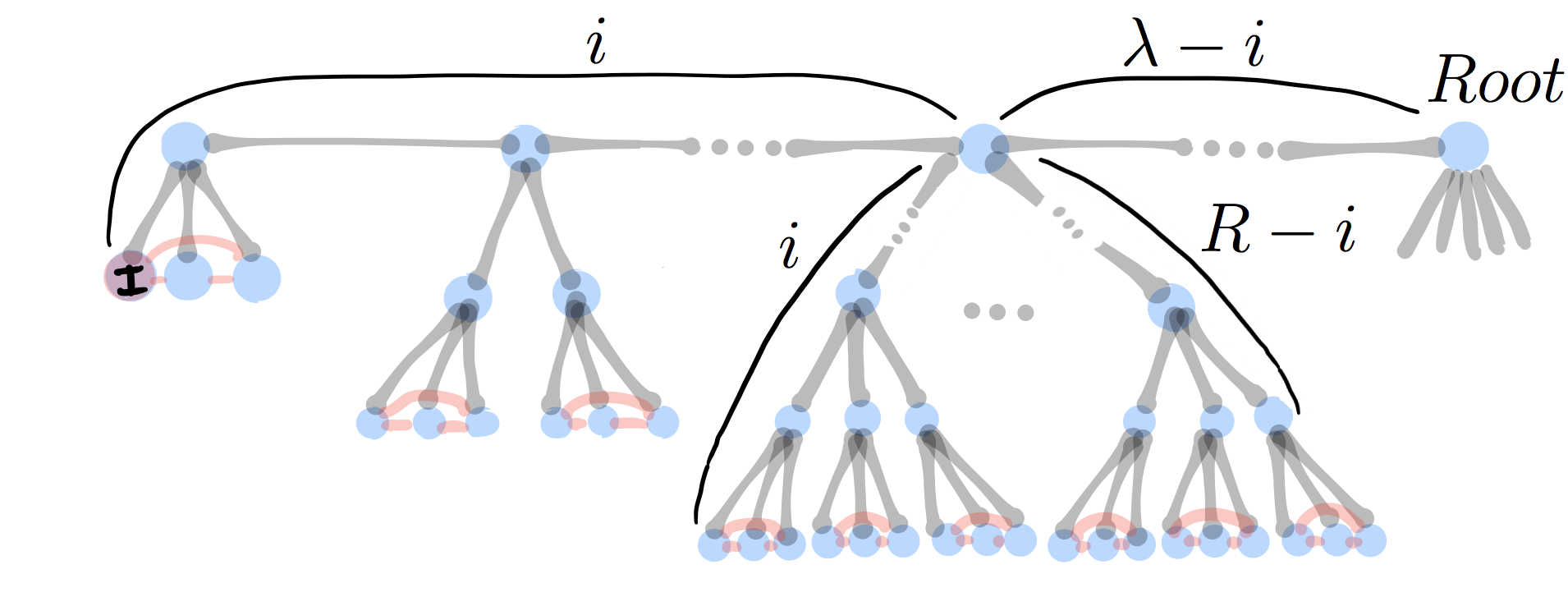}}}
\caption{The main branch of $I$.} \label{yo25}
\end{figure}

\begin{lemma}\label{lemma10}
For $x$ a vertex in $\mathcal{A}_{d}^{\lambda}$, one has the following:
\begin{eqnarray}
d(x,I)=\left\{\begin{array}{ll}
 |D_{x}|+\lambda&{~\rm~if~}|D_x\cap D_I|=0;\\
 |D_x|+\lambda-2\alpha&{~\rm~if~}|D_x\cap D_I|=\alpha>0.
\end{array}
\right.
\end{eqnarray}
\end{lemma}

\begin{proof}
Consider first the case of  $|D_x\cap D_I|=0$, for which one knows that $x$ is not in the main branch of the index vertex $I$, and hence to connect to the index node $I$ one needs to go through the root of the tree. Therefore the shortest path is given by $D_x\cup D_I$, and the first part of the lemma follows. 

For $x$ such that $|D_x\cap D_I|=\alpha>0$, there exists $v_0\in D_x\cap D_I$ such that $$\alpha=|D_{v_{0}}|\geq |D_v|$$ for all $v\in  D_x\cap D_I$. Thus,
 \begin{eqnarray}d(x,I)&=&d(x,v_0)+d(v_0,I)\\&=&\left(|D_x|-\alpha\right)+\left(\lambda-\alpha\right)\\&=&|D_x|+\lambda-2\alpha
 \end{eqnarray} and the lemma follows.
\end{proof}

In order to calculate the probabilities of infections lasting certain amount of time, and thus find the induced safe zones, one needs to understand the number of vertices at different distances from the infection index $I$.

\begin{proposition}\label{lemmalast}
The number $N_R$ of vertices $v$ of $\mathcal{A}_d^{\lambda}~$ for which $d(v,I)=R$ is 
\begin{eqnarray}
N_R=\left\{\begin{array}{ll}
(d-1)^{R-\lambda}+\sum_{R/2\leq i < \lambda}\eta(d,R,i)&R\geq\lambda;\\
1+\sum_{R/2\leq i < R} \eta(d,R,i)&R<\lambda,
\end{array}
\right.
\end{eqnarray}
for $\eta(d,R,i):=(d-2)(d-1)^{R-1-i}$.
\end{proposition}

  \begin{proof}
Consider first those nodes $v\in\mathcal{A}_d^\lambda$ for which $d(v,I)=R<\lambda$. By definition of $\mathcal{A}_{d}^\lambda$, and as seen in Figure \ref{yo25}, for each $R/2\leq i < R$ there are  
\begin{eqnarray}\label{new1}\eta(d,R,i)=(d-2)(d-1)^{R-1-i}\end{eqnarray}
vertices $x$ for which 
\begin{eqnarray}\label{new2}d(v,I)=R~{\rm ~and~}~|D_v\cap D_I|=\lambda-i.\end{eqnarray}
Moreover, for $i=R$ there is only one vertex. Indeed, from the above lemma those vertices $x$ satisfy  $d(v,I)=|D_v|-\lambda+2i=R,$ and thus the proposition  for $R<\lambda$ follows   summing over all possible $i$, where $R/2\leq i<R$.

In order to understand the case of $R\geq \lambda$, note that    those nodes $v\in\mathcal{A}_d^\lambda$ for which $d(v,I)=R\geq\lambda$   satisfy $| D_v\cap D_I | =\lambda-i$ for $0\leq i \leq \lambda$. These cases can be classified as follows:
\begin{itemize}
\item[(a)]  $| D_v\cap D_I | =0$, or 
\item[(b)]  $| D_v\cap D_I |= \lambda$, or

\item[(c)]  $| D_v\cap D_I |= \lambda - i  $ for $0< i<\lambda$.
\end{itemize}
When considering case (a), note that by Lemma \ref{lemma10} one needs to consider points $v$ in all branches not containing $I$ of $\mathcal{A}_d^{\lambda}$, such that $|D_v|= R-\lambda$.   By definition of a perfect $(d-1)-$ary tree,  there are  $(d-1)^{R-\lambda}$ such points.

Consider next the case (b) above, and note that when $i=0$, then there is only one vertex $v$  satisfying the condition, which is the root, and thus corresponds to $R=\lambda$. Finally, consider  case (c):  from \eqref{new1}-\eqref{new2} and Figure \ref{yo25}, for each $R/2\leq i < \lambda$, one has  $(d-2)(d-1)^{R-1-i}$ vertices $v$ for which $d(v,I)=R$ and $|D_x\cap D_I|=\lambda-i$. The proposition then follows for $R\geq \lambda$ summing over all cases (a), (b) and (c).
 \end{proof}

 The results from the above proposition can be visualised through Mathematica for a fixed value of $\lambda$ as in Figure \ref{distance}. 
 In the following sections we will use these results to study different models of epidemic outbreaks on $d$-cliqued tree networks. 
Finally, one should note that what has been shown in Lemma \ref{lemma10} and Proposition \ref{casesroot} serves to obtain the probabilities $P_i$ of an outbreak lasting $i$ ticks within $d$-ary trees when the index case is in the outermost vertex. These contact networks are those studied in  \cite{Seibold}, but such probabilities were not given by the authors. Because of this, we calculated them here since we needed them to derive several useful comparisons between our model and that of  \cite{Seibold}, and in particular to produce Figure \ref{like1} and Figure \ref{like}.

 \begin{figure}[h]
\centering
{\resizebox*{7cm}{!}{\includegraphics{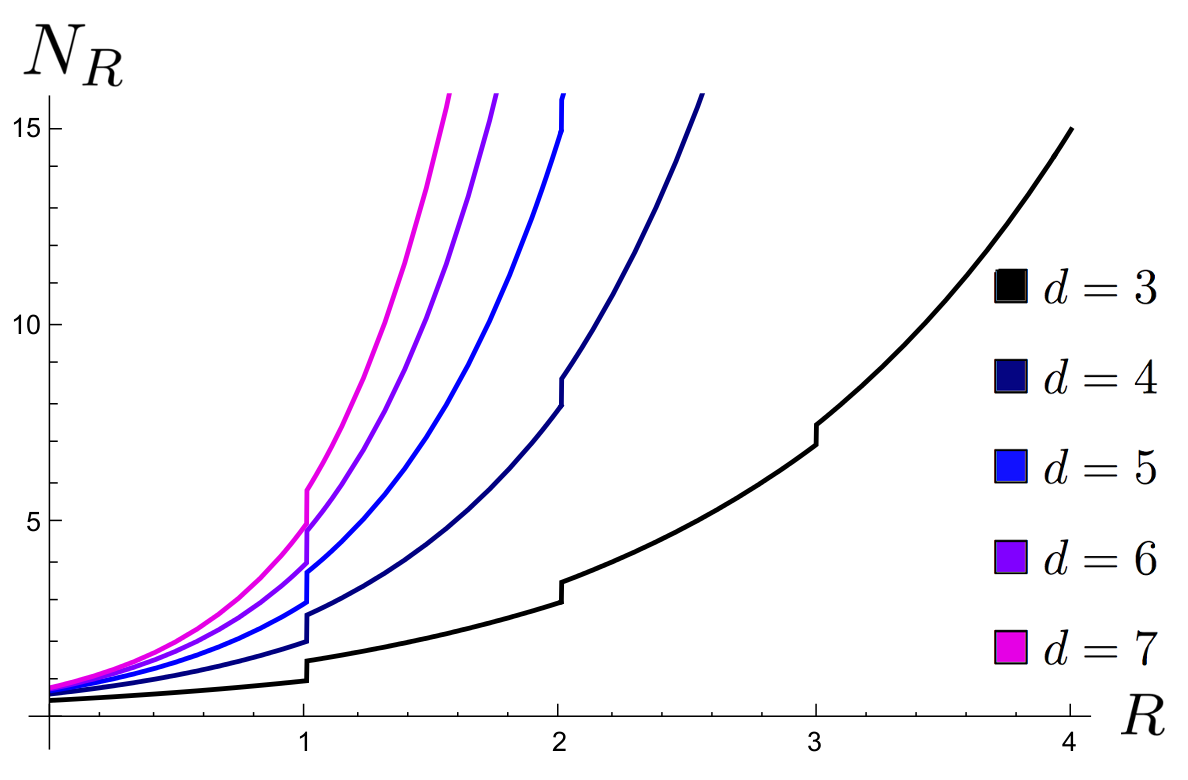}}}
\caption{The number $N_R$ of vertices at distance $R<\lambda$.} \label{distance} 
\end{figure}

\section{Duration of an outbreak using a next-generation model}
 
 Recall that in a next-generation model, one sets the probability of infection to $P_{inf}=1$, and the recovery probability to  $P_{rec}=1$.  Then, in this case, the outbreak duration is $2\lambda+1$, independently of the degree $d$. The examples in Figure \ref{holamas20} (a) for our model should be compared to \cite[Figure 6]{Seibold} appearing in  Figure \ref{holamas20} (b) for the model on perfect $d$-ary trees.
\begin{figure}[h]
\centering
\subfigure[~Disease duration in a next-generation model with $P_{inf} =1$ on $\mathcal{A}_{d}^{\lambda}$.]{%
\resizebox*{8.2cm}{!}{\includegraphics{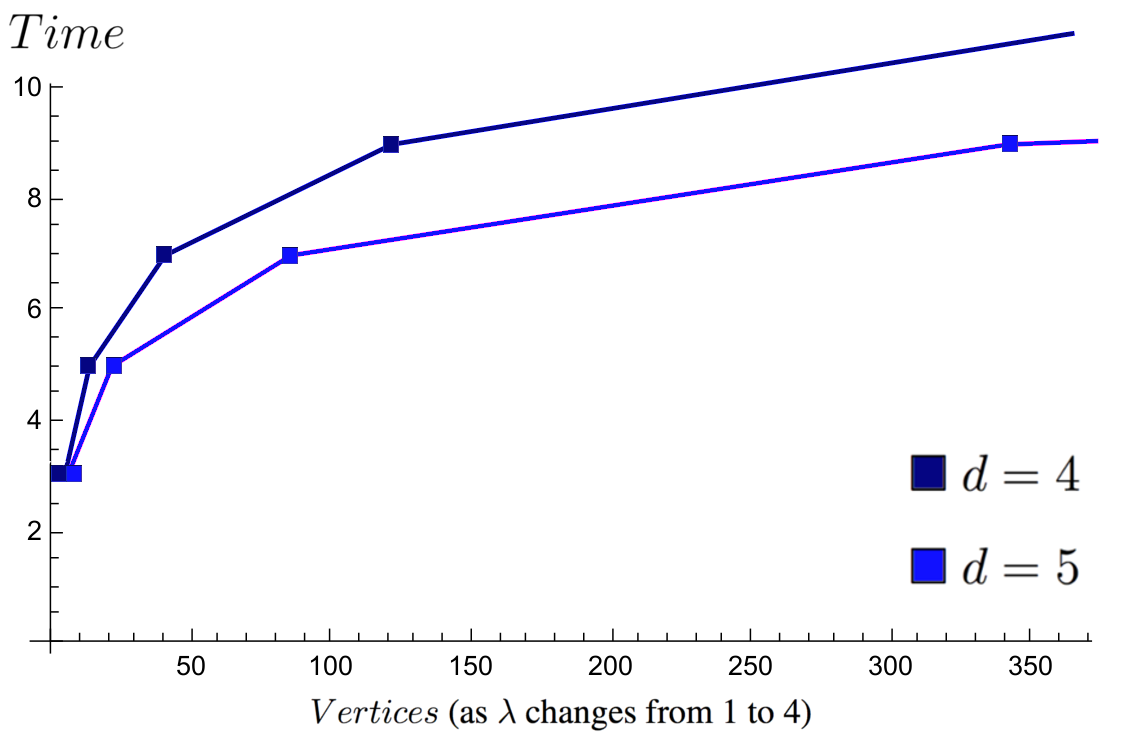}}}\hspace{5pt}
\subfigure[~Disease duration in $2$-ary trees from Figure 6 of  Seibold-Callender (2016) \cite{Seibold}.~]{%
\resizebox*{8.2cm}{!}{\includegraphics{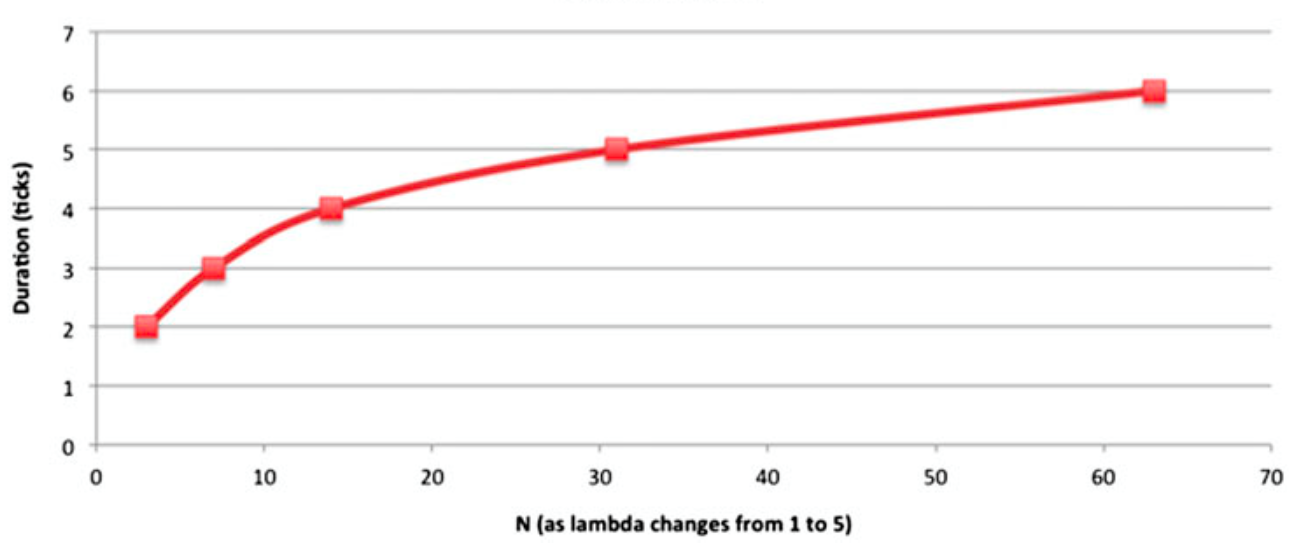}}}
\caption{Comparison of disease duration in a next-generation model.} \label{holamas20} 
\end{figure}

As seen before, the degree $d$ of $\mathcal{A}_d^{\lambda}$ determines the number of nodes of the graph once the  height $\lambda$ is fixed. However, for  SIR models on these graphs, just as in the case of  \cite{Seibold}, the degree $d$  has no effect on duration when the root is the index case, and only the height $\lambda$ matters. On the other hand, since the height determines the network size, one can see how it influences the duration of the outbreak as in Figure \ref{holamas20} (a) giving the duration scaling as $2\lambda+1$ with population size  stated as before.

 \subsection{Duration scaling with degree $d$ and  height $\lambda$}

From the previous analysis, when the index case is located in the root, the outbreak will have the shortest duration regardless of the value of $d$, since   the maximum distance from the root to any given vertex is the height  $\lambda$. Taking into account the additional +1 step required to have the last infected nodes to become removed one has that the duration of the outbreak needs to be at most $\lambda+1$. Note that, in particular, this case is equivalent to the one in  \cite{Seibold}. 

Equivalently, the work for d-ary trees will also hold when a random index node is chosen:  since the diameter of the graph $\mathcal{A}_d^{\lambda}$ is $2\lambda$, the duration is bounded above by $2\lambda+1$ when $d>1$. Note that when  $d = 1$, the graph $A_1^\lambda$ only exists for $\lambda=1$, and it consists of two vertices connected by one edge. Hence, in this case, a random index will be equivalent to having the index in the root. As done before, we will denote by $\delta_\lambda$ the expected average outbreak duration, and $P_i$ the probability the outbreak lasts exactly $i$  ticks. Then,  $\delta_{\lambda}=2~{\rm ~for~} d=1$ and the bounds of duration $\delta_\lambda$ otherwise are therefore given by
\begin{eqnarray}
 \lambda +1\leq \delta_{\lambda} \leq 2\lambda+1~{\rm ~when~}~ d>1. 
 \end{eqnarray}

 \subsection{Duration scaling with degree $d$ and the  index case location}
As mentioned before, the position of the origin of the infection (the index case $I$) within the graph will greatly determine the bounds on the outbreak durations. Moreover, both in {\bf Case (I)}, as well as in {\bf Case (II)}  of Section \ref{cases},  the outbreak durations for a next generation model will return results equivalent to those in  \cite{Seibold}.
To see this, one needs to keep in mind that in a next generation model, the duration will always be ${\rm max}\{d(v,I)+1~|~ v\in \mathcal{A}_d^{\lambda}\}.$
 
 As in the case of $d$-ary trees, when a random index case is introduced away from the outermost vertex, the mean duration is given by the weighted average of the duration for each node multiplied by the number of nodes with that duration, this is, the weighted average of the vertices in each level multiplied by the duration they would produce. Hence, for each value of $\lambda>1$, the work in  \cite{Seibold} still holds and  the expected outbreak duration $\delta_\lambda$ is
  \begin{eqnarray}
\delta_\lambda=\begin{array}{lc}
\frac{1}{\sum_{i=0}^{\lambda}d^i}\sum_{j=0}^{\lambda}\left((d^j)(\lambda+j+1)\right),\end{array}
\end{eqnarray} 
and is depicted in Figure \ref{yo255}: 
\begin{figure}[h]
\centering
\subfigure{%
\resizebox*{7cm}{!}{\includegraphics{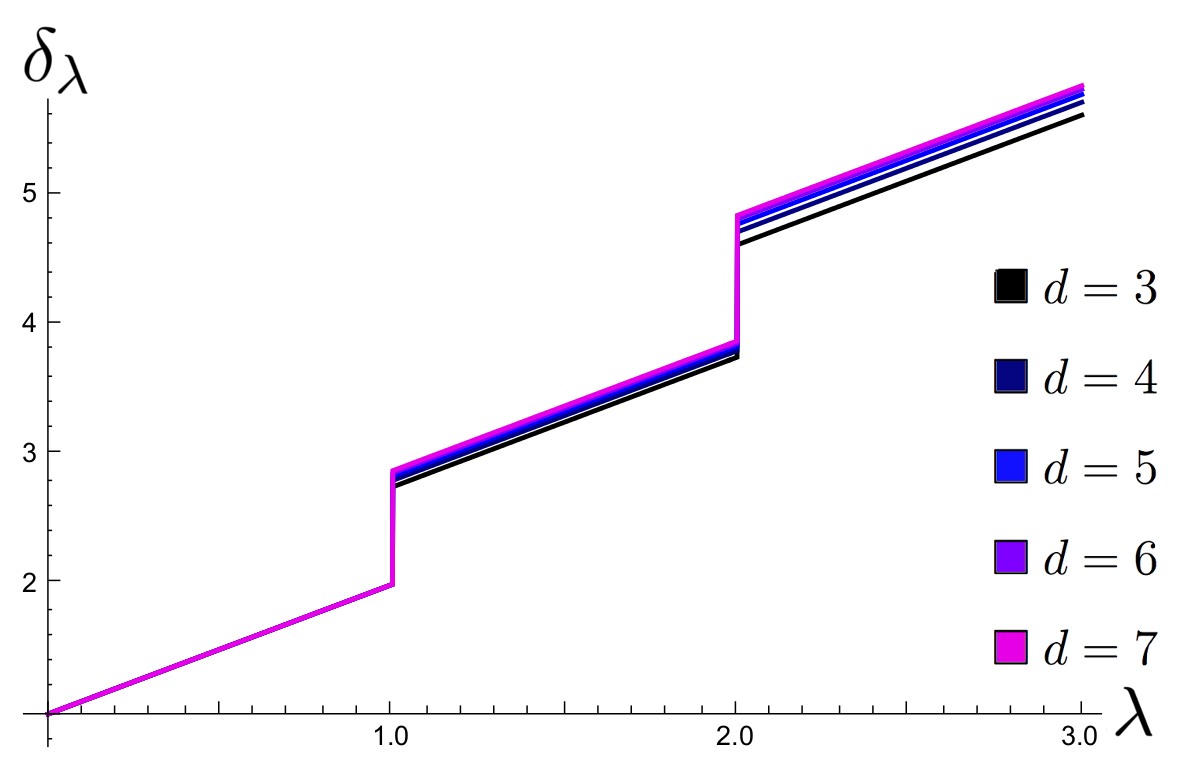}}}
\caption{The mean duration outbreak $\delta_\lambda$.} \label{yo255}
\end{figure}

When $\lambda=1$ the results are quite different. Indeed, in this case the graph $\mathcal{A}_{d}^1$ is the complete graph on $d+1$ vertices, and thus all nodes are infected within 1 tick, whatever the choice of infection index is. Hence, the duration outbreak in this case is 2,   different from Eq. (\ref{mean}) of  \cite{Seibold} for d-ary trees. Moreover, these differences become even more apparent when allowing the probabilities of infection  $P_{inf}$ and recovery $P_{rec}$ to vary, as in Figure \ref{difference}. 
 \begin{figure}[h]
\centering
\subfigure[The average outbreak duration  $\delta_{\lambda}$ depending on the recovery probability in $\mathcal{A}_4^{1}$, averages  taken over 20 simulations in IONTW.]{%
\resizebox*{7cm}{!}{\includegraphics{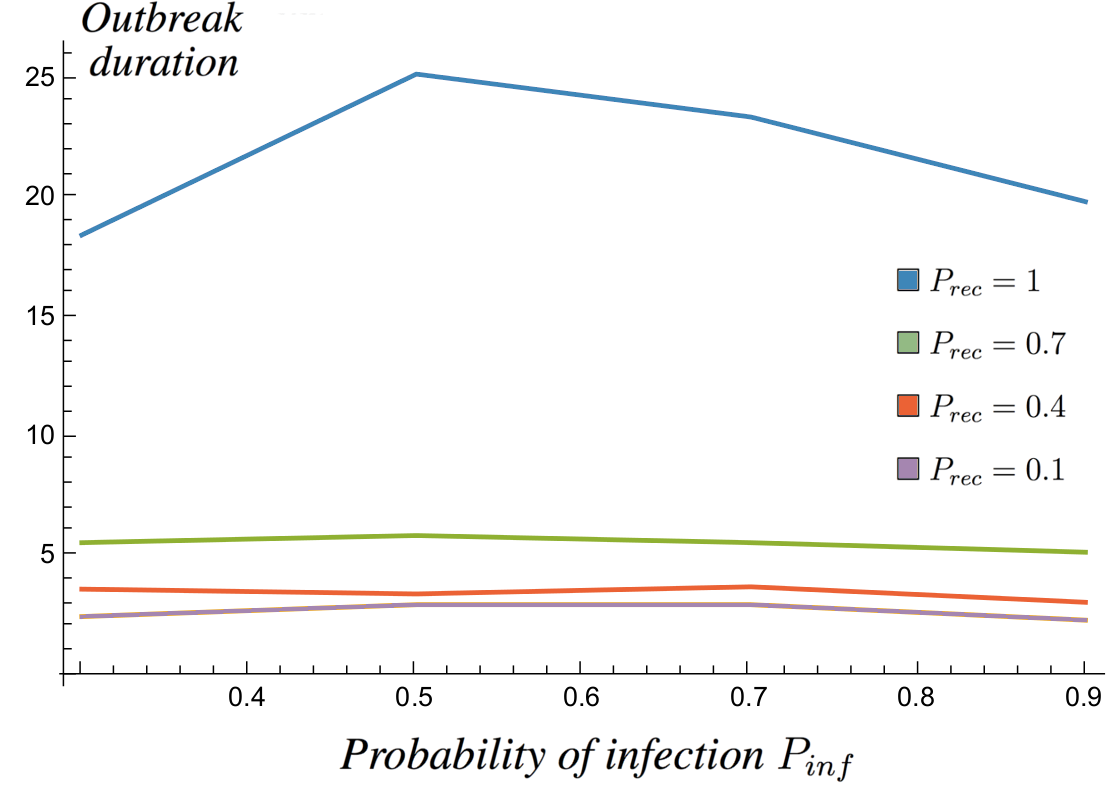}}}\hspace{5pt}
\subfigure[Duration  on a $4$-regular tree graph with $\lambda=1$ from  Seibold-Callender (2016) through IONTW.]{%
\resizebox*{7cm}{!}{\includegraphics{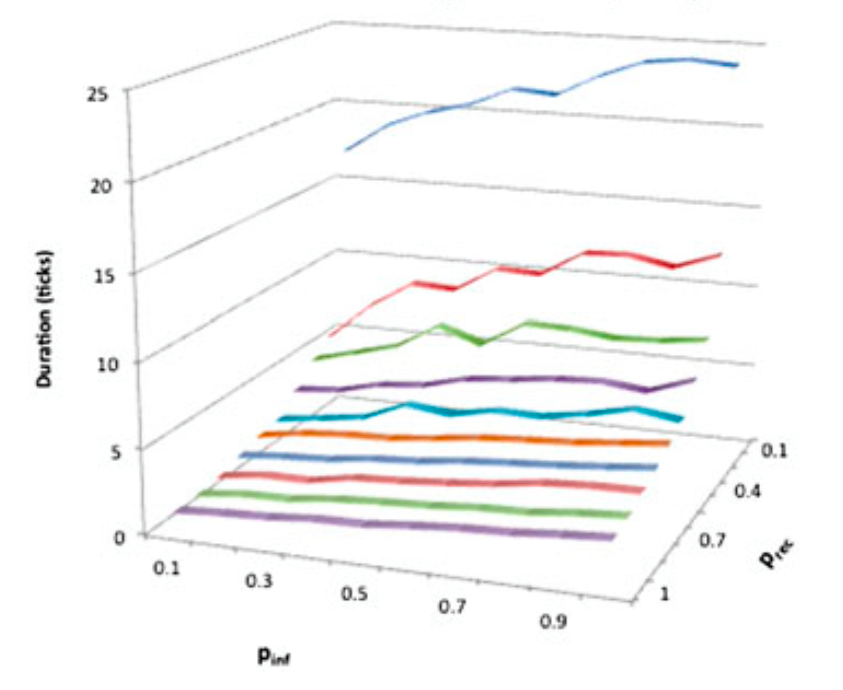}}}
\caption{Durations for   graphs with   $\lambda=1,$ and $d=4$.} \label{difference} 
\end{figure} 

 \vspace{-0.5 cm}

\section{Duration of an outbreak using a general discrete-time model}

When considering discrete-time models instead of next-generation models, outbreaks are specified in terms of varying   probabilities of infection $P_{inf}$ and of recovery $P_{rec}$, which can take values  between 0 and 1. As an example, if $P_{inf} = 0.25$, then the adjacent node(s) to an infectious node will have a $25\%$ probability of becoming infectious at the next time step. In particular, when $P_{inf}$ decreases, some vertices may escape infection and thus the infection outbreak might last less than $2\lambda+1$ ticks (recall that in  Figure \ref{difference} one could see the influence of the different probabilities on  the duration outbreak). In what follows we shall study the probabilities $P_j$ of an outbreak lasting $j$ tick (time units), and how these correlate to the scaling of different parameters.

\subsection{Duration scaling with population size}

Given an index node $I$ in an outermost vertex,   the values of the probabilities $P_j$ for $j<\lambda$ will be much higher than in  \cite{Seibold}. 
To calculate  the average duration  as a weighted average of all the possible durations, as well as to calculate  the probabilities $P_j$ for $ j\leq \lambda$, one needs to consider the chance of an outbreak lasting $j$ ticks and no one else getting infected after that time. For instance,   by construction, 
$P_1= (1-P_{inf})^{d-1}.$

 As a comparison with  \cite{Seibold},  consider an outbreak with the  index   case as one of the outermost vertices of $\mathcal{A}_3^{\lambda}$, this is, a vertex only adjacent to one of the $3$-clique's vertices. Setting the recovery probability $P_{rec}=1$ but leaving $P_{inf}<1$ unfixed,  there are two nodes connected to the infectious index node. Then, there are three possible options for infection at the next time step: either no node becomes infected (with probability $(1-P_{inf} )^2$), both nodes become infected with probability $P_{inf}^2$, or    only one node becomes infected. Since it could be either adjacent node,  this has probability of $2\cdot (P_{inf}(1-P_{inf})).$  
Through the same analysis as above, for   generic degree $d$, all possible combinations of infectious nodes need to be taken into account. Hence,   since the index node has $d-1$ adjacent vertices, in $\mathcal{A}_d^{\lambda}$ it can infect $k$ vertices in $\binom{d-1}{k}$ ways, each with probability $P_{inf}^k(1-P_{inf})^{d-1-k}.$ 

To calculate $P_2$ one would need the infection to last 2 ticks and no more, and thus all infected nodes at time 2 would not be allowed to infect anyone else. The calculation becomes more complicated in this setting than in $d$-ary trees since nodes which were not infected at time $j$ within the clique might become infected at time $j+1$ through a non-minimal path to the index node. Hence, when calculating the total probabilities, one needs to differentiate between whether the vertices connected to the body of the clique are infected or not. In what follows we shall illustrate how to calculate these probabilities. 

Recall that the index case is adjacent to $d-2$ vertices of the clique who are only adjacent to other vertices of the clique, and one vertex $v_0$ of the clique adjacent to the body of $\mathcal{A}_d^{\lambda}$,  each which may be infected with probability $P_{inf}$, or not with probability $(1-P_{inf})$. If $v_0$ is not infected, then to calculate $P_2$, it should remain not infected at time 3, which will happen with probability
 \begin{eqnarray}
&~&\underbrace{(1-P_{inf})^2}_{v_0~{\rm uninfected~at~time~2~and~3}}\nonumber\\&~&\left(
 \sum_{k=1}^{d-2}\underbrace{\binom{d-2}{k} P_{inf}^{k} (1-P_{inf})^{d-2-k}}_{{\rm only~}k~{\rm infected~at~time~2}}\right.\nonumber\\ &~&\left.\overbrace{(1-P_{inf})^{d-2-k}}^{\rm no~uninfected ~get~infected~at~time~3}\right).\label{p21}
 \end{eqnarray}
  
  When the vertex $v_0$ is infected, something similar happens to make the outbreak last 2 ticks, except that the neighbours of $v_0$ outside the clique need to be taken into account. Hence, this case  contributes towards $P_2$ with the following 
 \begin{eqnarray} 
&~&\underbrace{P_{inf}(1-P_{inf})}_{v_0~{\rm does~not~infect~outside~clique}} \nonumber \\ &~& \left(
 \sum_{k=1}^{d-2}\underbrace{\binom{d-2}{k} P_{inf}^{k} (1-P_{inf})^{d-2-k}}_{{\rm only~}k~{\rm infected~at~time~2}}\right.  \nonumber\\ &~& \left.\overbrace{(1-P_{inf})^{d-2-k}}^{\rm no~uninfected ~get~infected~at~time~3}\right)  \label{p22}\end{eqnarray}

 Then, to calculate $P_2$ one needs to sum  Eq. \eqref{p21} and  Eq. \eqref{p22}, leading to
 
  \begin{eqnarray}P_2&=&
 \left((1-P_{inf})^2+ P_{inf}(1-P_{inf}) \right)\nonumber\\&~&\left(
 \sum_{k=1}^{d-2} \binom{d-2}{k} P_{inf}^{k} (1-P_{inf})^{2d-4-2k}  \right)\nonumber\\
 &=&  \sum_{k=1}^{d-2} \binom{d-2}{k} P_{inf}^{k} (1-P_{inf})^{2d-3-2k}. \label{pdA}
 \end{eqnarray}

 Whilst one could compare this probability with the one found in  \cite{Seibold}, the latter was obtained for an index case in the root and thus the comparison would not be too enlightening (see Figure \ref{differenceP2} (a)). However, we can deduce what the probabilities would have been for an index case in the same position as the one we are taking in $\mathcal{A}_d^\lambda$, in a vertex only adjacent to clique members.

  As an example, in the setting of  \cite{Seibold}, which is for $d$-ary trees, one has that when $d=4$ the probability of the outbreak lasting 2 ticks is
 \begin{eqnarray}
 P_2=P_{inf}(1-P_{inf})^4.\label{P21}
 \end{eqnarray}
 Moreover, one may want to also compare this to the setting of $d$-regular trees, in which case for $d=4$ one has 
  \begin{eqnarray}
 P_2=P_{inf}(1-P_{inf})^3.\label{P22}
 \end{eqnarray}
 Comparing all of these settings, one can see how the chances of the outbreak lasting 2 ticks are much higher in the case of $\mathcal{A}_d^{\lambda}$, as seen in Figure \ref{differenceP2} (b) below.

  \begin{figure}[h]
\centering
\subfigure[Root index  of $4$-ary tree as in  Eq. \eqref{p2tree} from    \cite{Seibold}, compared with outermost index in $\mathcal{A}_4^{\lambda}$.]{%
\resizebox*{6 cm}{!}{\includegraphics{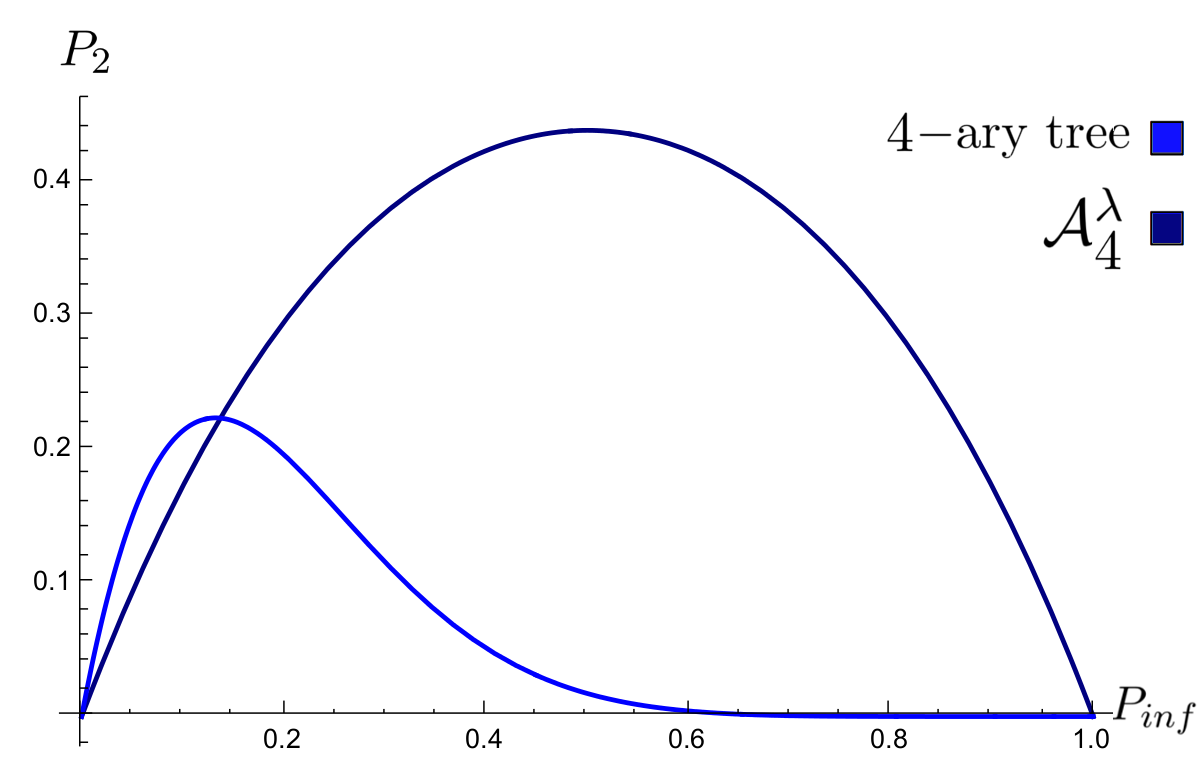}}}\hspace{5pt}
\subfigure[Index case in outermost vertex of the $4$-ary tree as in Eq. \eqref{P21}, and of the $4$-regular tree as in  Eq. \eqref{P22}, compared to $\mathcal{A}_4^{\lambda}$.]{%
\resizebox*{6 cm}{!}{\includegraphics{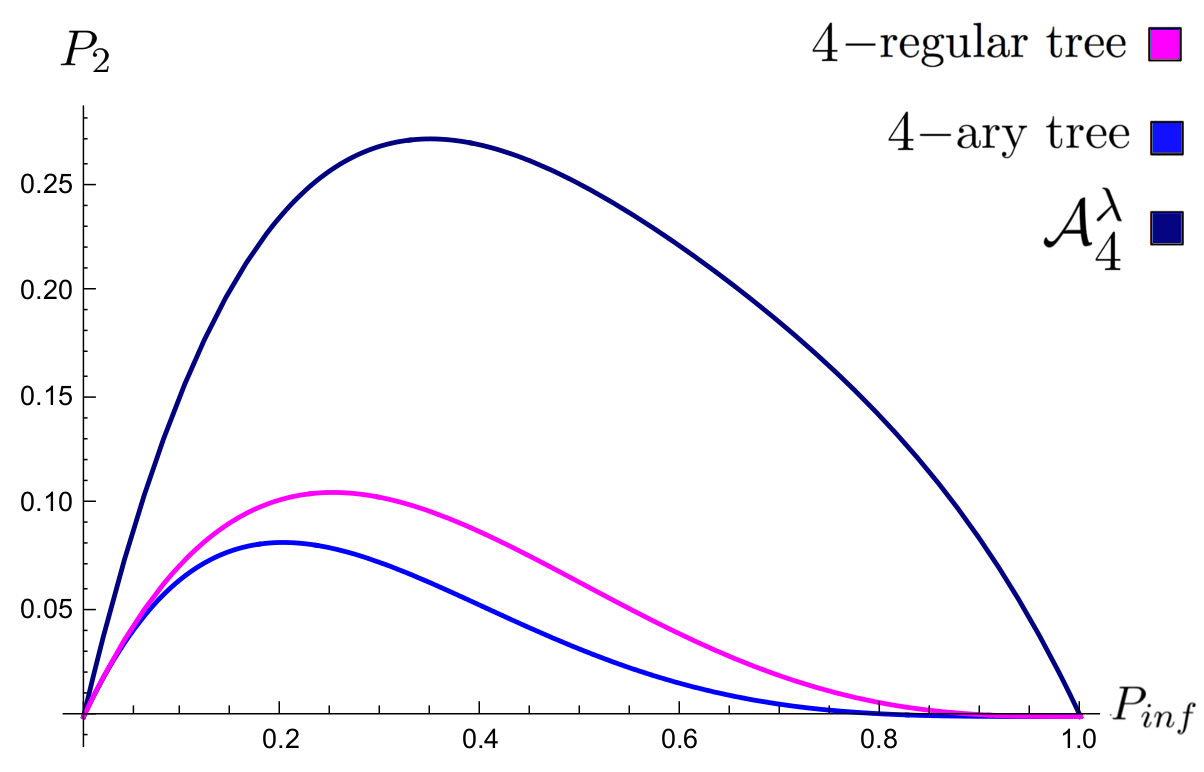}}}
\caption{The values of $P_2$ in terms of the probability of infection $P_{inf}$ for the graph $\mathcal{A}_{4}^{\lambda}$ of  Eq. \eqref{pdA},  and different index cases for the $4$-ary tree of  \cite{Seibold} , and $d$-regular trees.} \label{differenceP2} 
\end{figure}
 
\subsection{Duration scaling with degree $d$}
Whilst the height $\lambda$ of the graph $\mathcal{A}_{d}^{\lambda}$ will not influence the values of $P_i$ for $i<\lambda$, the variation of  the degree $d$ will have a significant effect. In particular, it will enhance the differences between our model and the one studied in  \cite{Seibold}. Consider for instance $P_{inf}=0.4$. Then, for different values of the degree $d$, the probability $P_2$ for $d$-ary trees and for $d$-regular trees can be compared to the one for our $d$-cliqued networks $\mathcal{A}_d^{\lambda}$   in Figure \ref{P222}.

  \begin{figure}[h]
\centering
\subfigure[Probability $P_2$ for $d$-ary trees and $d$-reguar trees.]{%
\resizebox*{7cm}{!}{\includegraphics{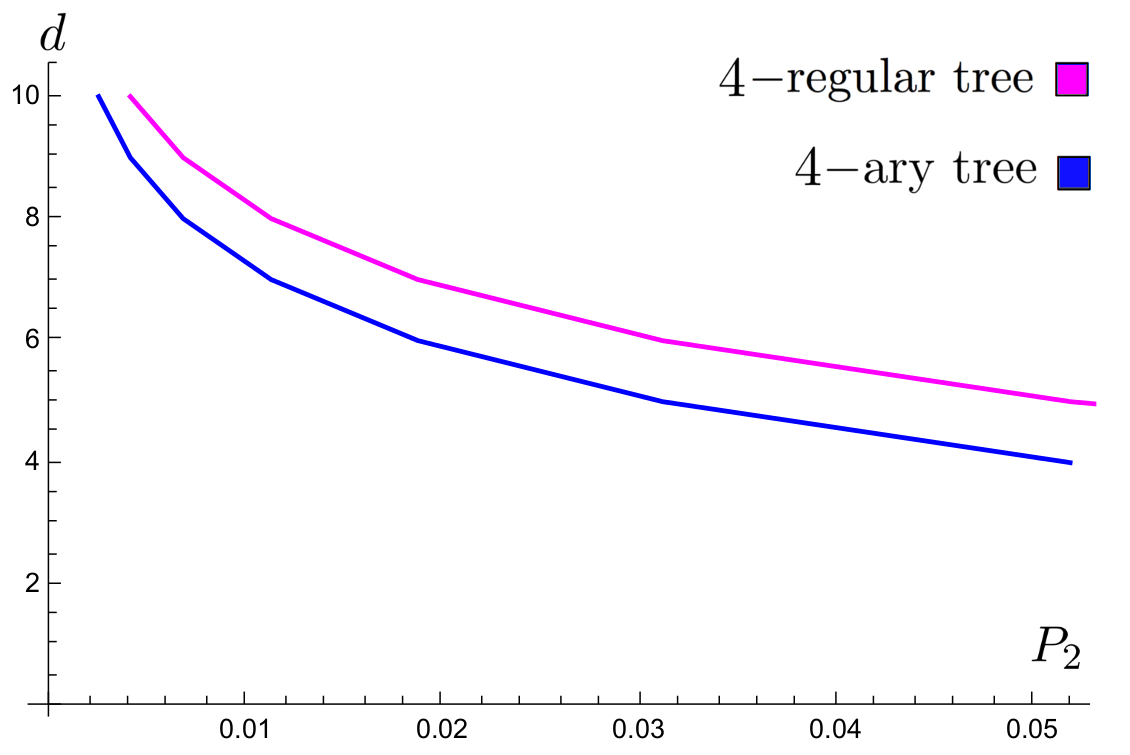}}}\hspace{5pt}
\subfigure[Probability $P_2$ for $\mathcal{A}_{d}^{\lambda}$.]{%
\resizebox*{7cm}{!}{\includegraphics{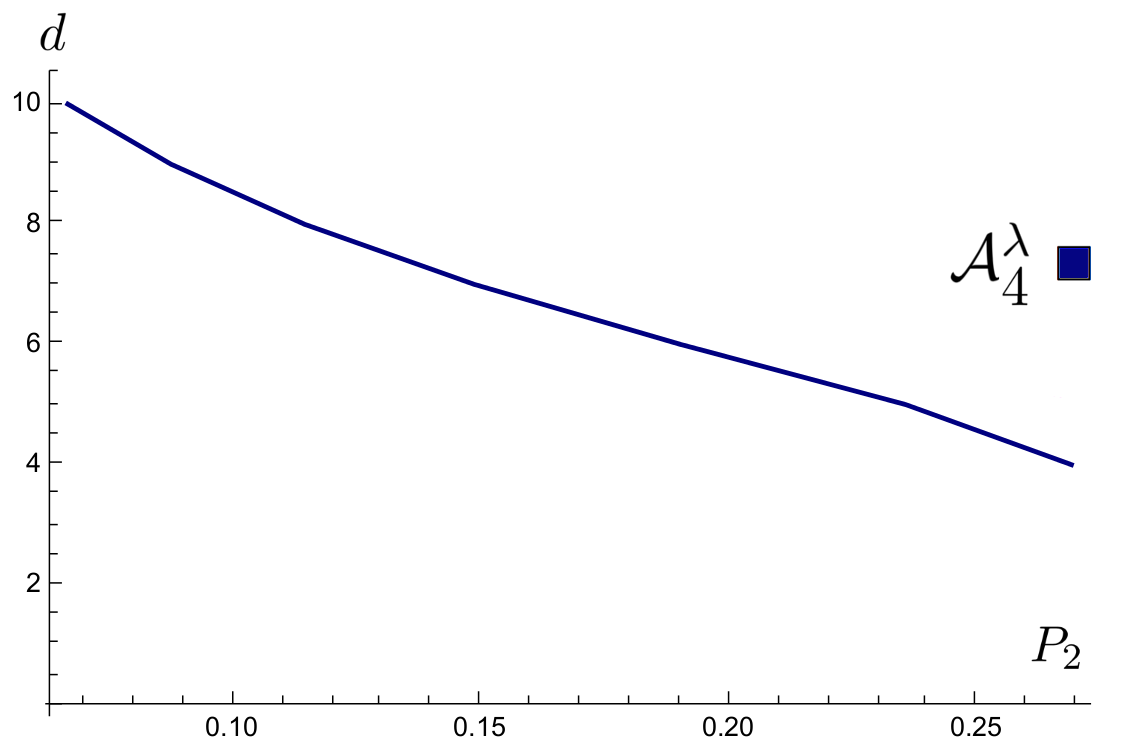}}}
\caption{The probability $P_2$ that an outbreak lasts 2 ticks for our model on $d$-cliqued networks and for the model in   \cite{Seibold}, for index node in outermost vertex.} \label{P222} 
\end{figure}

Although we have focused previously on the probability $P_2$, similar analysis leads to equivalent results for $P_i$ for higher values of $i$. 
Each probability $P_i$ will be given by a nested sum of $i-1$ terms. The outermost sum refers to the nodes directly adjacent to the index case inside the clique (which also appeared in $P_2$). The innermost sum refers to the nodes that experience infection and are at a distance of $i-1$ from the root, and thus these probabilities can be deduced through Proposition \ref{lemmalast}, as we shall see below. To this end, recall that for $\mathcal{A}_d^{\lambda}$ one has
  \begin{eqnarray} P_1&=&(1-P_{inf})^{d-1};   \nonumber\\ P_2&=&
 \sum_{k=1}^{d-2} \binom{d-2}{k} P_{inf}^{k} (1-P_{inf})^{2d-3-2k}. \nonumber
 \end{eqnarray}
 Then, for each $3\leq i \leq d-1$ one needs to separate the study of $P_i$ into $2^i$ cases depending on weather the vertex  $v_0$ in the index's clique adjacent to the body  of $\mathcal{A}_d^\lambda$ is infected or not at each step, as done for the case of $P_2$. 
 
 To illustrate the technique, we shall show here how these calculations are done for $P_3$. In this case, if the vertex $v_0$ is not infected at time 2 or 3, then the probability $P_3$ is given by  Eq. \eqref{mas11}.

When $v_0$ is infected at time 3, then the probability $P_3$ is similar to the one before, except that the only vertex in the body of $A_d^\lambda$ which is adjacent to $v_0$ need not get infected in the next step. Hence, in this case  through the term $\star$ {\bf above} one obtains a contribution of
\begin{eqnarray}
\underbrace{P_{inf}(1-P_{inf})^3}_{v_0~{\rm ~infected~ at~ time~ 3}} \cdot ~\star \label{masmas}
\end{eqnarray}
If the vertex $v_0$ was infected at time 2, then it could infect its adjacent body vertex or not, and this will give two different cases equivalent to the cases in  Eq. \eqref{p22} for $P_2$. If no one in the body is infected, then one has that the contribution to $P_3$ is as in  Eq. \eqref{masmas}. On the other hand, if $v_0$ infects someone at time 3, then that vertex needs not to infect anyone at time $4$. Since this vertex has $d-1$ adjacent nodes (e.g. see Figure \ref{yo25}), one has that in this case the contribution to  $P_3$ is given by 
\begin{eqnarray}\underbrace{P_{inf}^2(1-P_{inf})}_{v_0~{\rm  infects ~body~vertex}}~\cdot ~\overbrace{(1-P_{inf})^{d-1}}^{{\rm no~further~body~vertex~infected}}~ \cdot ~\star\nonumber \end{eqnarray}
Summing over   these four cases, one has 
 \begin{eqnarray}
P_3&=&\left[(1-P_{inf})^3+2P_{inf}(1-P_{inf})^3\right.\nonumber \\&~&\left.+P_{inf}^2(1-P_{inf})(1-P_{inf})^{d-1}\right]~ \cdot ~\star\nonumber\\
&=&\left[(1+2P_{inf}+P_{inf}^2)(1-P_{inf})^3\right.\nonumber \\&~&\left.+(1-P_{inf})^{d-3}\right]~ \cdot ~\star\nonumber
\end{eqnarray}
 
Therefore, one has that 
 \begin{eqnarray}\label{masmas22}
P_3= \left[(1-P^2_{inf})^2+(1-P_{inf})^{d-2}\right]~ \cdot ~\star
\end{eqnarray}
  In order to compare the probability of an outbreak lasting 2 times for a network modeled through  \cite{Seibold} and for a $d$-cliqued network, we consider the ratio between $P_2$ for the latter by the former as seen in Figure \ref{like}.  In particular, one can see that the higher the probability of infection $P_{inf}$, the higher is the probability $P_2$ for $\mathcal{A}_2^{\lambda}$ compared to the probability $P_2$ for $d$-ary trees with the same index root, and the same will apply for the  other $P_i$.

\vspace{-0.5 cm}
 \begin{widetext}
  \begin{eqnarray}
\underbrace{(1-P_{inf})^3}_{v_0~{\rm uninfected}}\overbrace{\left(
 \sum_{k=1}^{d-2}\underbrace{\binom{d-2}{k} P_{inf}^{k} (1-P_{inf})^{d-2-k}}_{{\rm only~}k~{\rm infected~at~time~2}}\underbrace{\left[\sum_{j=1}^{d-2-k}\binom{d-2-k}{j}P_{inf}^j \overbrace{ (1-P_{inf})^{2(d-2-k-j)}}^{{\rm no~more~infected}}\right]}_{j{\rm~get~infected~at~time~3}} 
\right).}^{\star}\label{mas11}
 \end{eqnarray}
  \end{widetext}
   \vspace{2 cm}

 \begin{figure}[h]
\centering
\subfigure{%
\resizebox*{8cm}{!}{\includegraphics{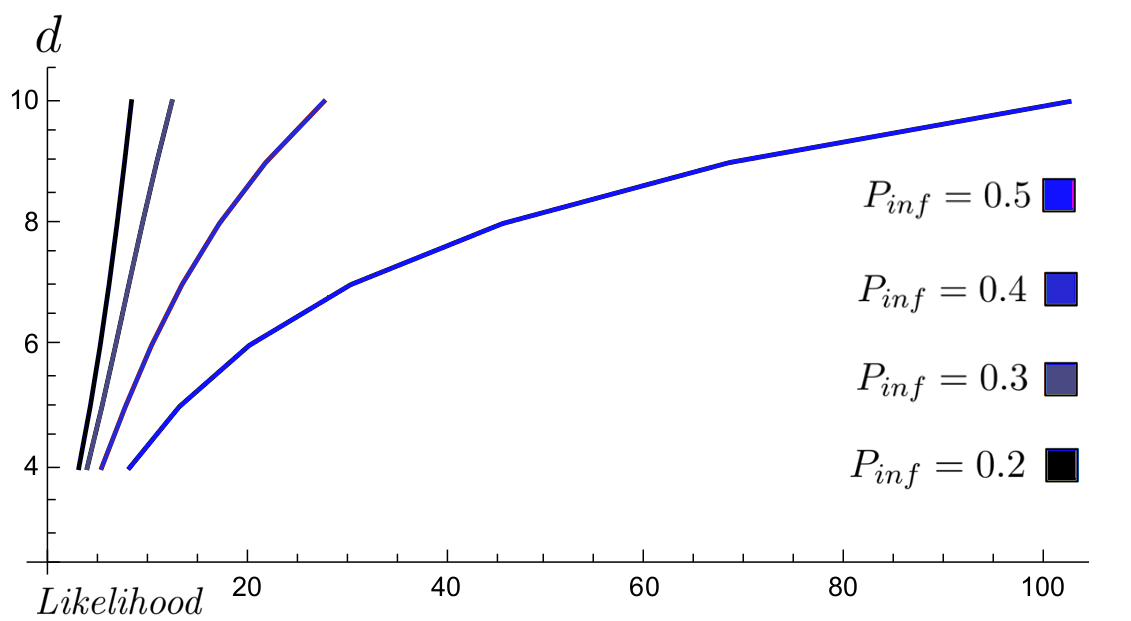}}}
\caption{The likelihood of $P_2$ depending on the degree $d$ for different values of $P_{inf}$. This is, the ratio of $P_2$ for $\mathcal{A}_d^{\lambda}$ and $P_2$ for $d$-ary trees, both with index vertex in outermost layer.} \label{like}
\end{figure}
\section{Future work} 

The work presented here introduces a new perspective to study epidemics in networks with complete sub-networks, and thus we shall mention here a few of the  many paths through which one may generalize our results.

\subsection{Safe zones in $d$-cliqued networks}\label{safe}
In order to contain an infection outbreak, one would like to understand how the removal of edges or vertices of a graph creates a safe zone of vertices which have a very low probability of an infection.  For this, we introduce the following definition: 

 \begin{definition}
 Given $1>>\varepsilon>0$, a \textbf{safe zone} $\mathcal{S}_\varepsilon^t$ of a graph $G$ is a connected subset of vertices of $G$ that have probability of infection at most $\varepsilon$ at time $t$. The {\bf size} $|\mathcal{S}_\varepsilon^t|$ is  the number of vertices in it. 
 \end{definition}

The structure of the safe zone depends on the shortest paths from $I$ to other vertices, as well as on the probability of infection $P_{inf}$, and of recovery $P_{rec}$.
The paths between vertices can be understood via  Lemma \ref{lemma10} and Proposition \ref{lemmalast}, and are closely related to the correspondence between the number of edges in $\mathcal{A}_{d}^\lambda$,   the degree $d$ and height $\lambda$, which from  Eq. \eqref{edgenumber}    can be seen in Figure \ref{holamas2}. 
We strive to find the number of  edges that can be removed in the contact network, depending on the nodes which are initially infected.

  We have seen that in many cases the outbreak duration in our networks will be quite different from the duration found in   \cite{Seibold}. In particular, the case of $\mathcal{A}_4^1$ appearing in  Figure   \ref{difference}   should be compared to the case of $\lambda=1,d=4$ of \cite[Figure 12]{Seibold} appearing in  Figure \ref{difference}.  
Moreover, when $1\leq i\leq d-1$ the probabilities $P_i$ of an outbreak lasting $i$ ticks increase quite drastically, since in  $\mathcal{A}_{d}^\lambda$ our index node $I$ is within a clique, and thus all vertices not infected in one step may still be infected in the other.  Hence, the safe zones for $\mathcal{A}_d^\lambda$ will be smaller than those   for perfect $d$-ary trees of  \cite{Seibold}.  For instance in  $ \mathcal{A}_{3}^2$ with $P_{inf}=1$ there is a safe zone with $|\mathcal{S}_\varepsilon^1|= 7$,   two with  $|\mathcal{S}_\varepsilon^2|=3$, and two with $|\mathcal{S}_\varepsilon^3|=2$.

\subsection{Epidemics  in random networks} Understanding epidemics on random networks is an open problem which has attracted attention from researchers in many areas of mathematics. For this, it would be most important to understand connectivity of the random networks under consideration, as well as their clique density.  One should note that whilst connectivity is important, the understanding of {\it safe zones} becomes much harder for connected graphs than for sparse graphs.  In particular, we expect to understand whether having many cliques preclude the existence of safe zones in random networks   
 
 \section{Conclusion}

 The importance of precise and efficient mathematical modeling of epidemics in order to analyze biological and social networks is widely accepted. As these systems are often symmetric, our results mostly focus on contact networks with symmetric clusters (cliques),   providing an enlightening setting from which to study society.   We have introduced several new mathematical objects which appear to be of importance when studying epidemic dynamics in those social groups  that have complete sub-networks, describing subgroups of people who are all related to each other. 
\smallbreak
In order to account for regular social groups where  interactions are allowed amongst the youngest members of the network, we introduced   a generalization of regular trees to {\it $d$-cliqued tree graphs} $\mathcal{A}_d^\lambda$ (see Definition \ref{defi1}).  These objects allow us to model, for example, social groups where children are connected to all their school class mates and not only their teachers. 
 \smallbreak

 In order to study infections in these new  graphs, we extended the results of    \cite{Seibold}  to our setting. We considered different ways in which an outbreak could originate, and showed that when it originated in the root of the graph the  results of  \cite{Seibold}  still hold. On the other hand, we proved  that when the  index $I$ of the graph  (the origin of the infection) is one of the youngest nodes (the red node in Figure \ref{casesroot} (c)), then its siblings will be infected immediately, while in the setting of   \cite{Seibold}, they wouldn't.

 Without the use of the  IONTW platform, we were able to study the characteristics of $d$-cliqued tree graphs  equivalent to those for $d$-ary trees of  \cite{Seibold}, and showed that in many cases safe zones are smaller in our setting.
Moreover, when the index vertex is within a network's clique and not adjacent to the body of the network, though the probabilities $P_i$ of an outbreak lasting $i$ ticks one could calculate the  expected average duration $\delta_{d}^{\lambda}$ of   the graph $\mathcal{A}_d^{\lambda}$ as the weighted average of each of the possible durations $
\delta_d^\lambda=\sum_{i=1}^{2\lambda+1}iP_i$. 
\smallbreak

Whilst in  \cite{Seibold} the model on perfect $d$-ary trees was shown to be greatly influenced by  the size of the graph when determining the duration of an outbreak, even more so than the number $d$ of branches which has minimal effect after the root is infected, we have shown that in our model the height is irrelevant when calculating $P_i$ for small values of $i$. In fact, we have described how these values depend mainly on the degree $d$ of the graph, which emphasises the difference between the two models.  

When considering  next generation models where $P_{inf} = P_{rec} = 1$, we saw that not much analysis was needed. However, interesting mathematics appeared when   implementing more general discrete models for the index case being in the youngest node, through which the different values of $P_{inf}$ influenced  the duration of the outbreak. One should note that even for this type of index, the probabilities $P_i$ for the $d$-ary trees were not given in   \cite{Seibold}, and thus our methods   help, in particular, to understand the simpler model of $d$-ary trees.  Since the distance between the index case and other vertices are a very important factor for both scenarios, we dedicated a section to give formulas for the number of nodes at each distance from the index case, which are then used to calculate each probability $P_i$ of an outbreak lasting exactly $i$ ticks or time units, both for the setting of  \cite{Seibold} as well as for ours.

Finally, by considering the ratio between $P_i$ for $d$-ary trees and for our $\mathcal{A}_d^\lambda$, we were able to quantise the different likelihoods of an outbreak lasting $P_i$ in the two different models.  In particular, as mentioned before, the higher  the probability of infection $P_{inf}$ is, the higher is the probability $P_2$ for $\mathcal{A}_d^{\lambda}$ compared to the probability $P_2$ for $d$-arey trees with the same index root, and equivalent results apply for other $P_i$. In other words, the chance of an outbreak lasting two time units is much higher (sometimes 100 times higher) for our new graphs, showing the importance of accounting for relations between multiple individuals (see Figure \ref{like}). \\

\section*{Acknowledgement(s)}

Much of the research in this paper was conducted by the second author under the supervision of the first one, as part of the MIT PRIMES-USA program. Both authors would like to thank  MIT PRIMES-USA for the opportunity to conduct this research together, and in particular Tanya Khovanova for her continued support. The authors are also thankful  to Kevin Ren, Fidel I.~Schaposnik and James A.~Unwin for their comments on the paper, and the referees of the journal for their useful corrections.

\section*{Funding}

The work of LPS is partially supported by NSF DMS-1509693, and by an Alexander von Humboldt Fellowship.


\begin{thebibliography}{99}

  
   
 \bibitem[Just-Callender-Lamar (2015)]{Just}  W.~Just, H.L.~Callender, D.M.~Lamar,  {\it Exploring Transmission of Infectious Diseases on Networks with NetLogo,} http://www.ohio.edu/people/just/ IONTW/,  (2015).
 
 
 \bibitem[Keeling \& Eames (2005)]{Keeling} M.J. Keeling, K.T. Eames, {\it Networks and epidemic models,} Journal of the Royal Society Interface, 2(4), 295--307, (2005).

\bibitem[Kermack \& McKendrick (1927)]{Kermack} W. O. Kermack, A.G. McKendrick, {\it A contribution to the mathematical theory of epidemics},
 {Proceedings of the Royal Society of London A: mathematical, physical and engineering sciences},
  {115},
   {772},
 {p. 700--721},
 {(1927)}.


 
 
 \bibitem[Pastor-Satorras \& Vespignani (2001)]{Pastor} R.~Pastor-Satorras, A.~Vespignani, {\it Epidemic dynamics and endemic states in complex networks}, Phys. Rev. E 63 (2001).
 
 \bibitem[Brauer \& Castillo-Chavez (2001)]{Reed} F. Brauer, C. Castillo-Chavez, {\it Mathematical Models in Population Biology and Epidemiology.} Springer (2001).
 
 \bibitem[Riley et. al. (2003)]{Riley} Riley S., Fraser, C., Donnelly, C. A., Ghani, A. C., Abu-Raddad, L. J., Hedley, A. J., Anderson, R. M. (2003). {\it Transmission dynamics of the etiological agent of SARS in Hong Kong: Impact of public health interventions}. Science, 300, 1961--1966.  
 
 \bibitem[Ross-Mandal-Sarkar-Sinha (2011)] {Ross} S. Mandal, R. R. Sarkar, S. Sinha, {\it Mathematical models of malaria - a review}, Malaria Journal 10, p. 202, (2011).
  
  \bibitem[Seibold-Callender(2016)]{Seibold} 
C.~Seibold,  H.L.~Callender. 2016.
``Modeling epidemics on a regular tree graph'',
 Letters in Biomathematics,
Vol. 3 No. 1, 
 p. 59--74.
   

 
  \bibitem[Wilensky (1999)]{Wilensky} U.~Wilensky, {\it Netlogo. Center for Connected Learningand Computer-Based Modeling}, Northwestern University: Evanston, IL.  (1999).

 
 \end{thebibliography}
\end{document}